\documentclass[11pt]{article}
\usepackage[dvipsnames,usenames]{xcolor}
\pdfoutput=1
\usepackage[utf8]{inputenc}
\usepackage[T1]{fontenc}
\usepackage[margin=1in]{geometry}
\usepackage{amsmath,amsthm,amssymb,thm-restate,booktabs,doi,graphicx,latexsym,url,xspace,bm}
\usepackage{microtype,hyperref}
\usepackage{cleveref}
\usepackage[ruled,vlined]{algorithm2e}
\usepackage[numbers,sort&compress]{natbib}

 \usepackage[inline]{enumitem}
\setlist[enumerate]{label=(\arabic*),labelindent=\parindent,leftmargin=*}

\usepackage{sectsty}
\allsectionsfont{\boldmath}

\definecolor{citecolor}{HTML}{0000C0}
\definecolor{urlcolor}{HTML}{000080}

\hypersetup{
    colorlinks=true,
    linkcolor=black,
    citecolor=citecolor,
    filecolor=black,
    urlcolor=urlcolor,
    pdftitle={}, %
    pdfauthor={} %
}

\usepackage{tikz}

\theoremstyle{plain}
\newtheorem{theorem}{Theorem}[section]
\RenewCommandCopy{\theHtheorem}{\thetheorem}
\newtheorem{lemma}[theorem]{Lemma}
\newtheorem{proposition}[theorem]{Proposition}
\newtheorem{corollary}[theorem]{Corollary}
\theoremstyle{definition} %
\newtheorem{example}[theorem]{Example}

\newcounter{claim}
\renewcommand{\theclaim}{\Alph{claim}}
\newenvironment{claim}{\refstepcounter{claim}%
\par\medskip\par\noindent{\it Claim~\theclaim.~}~\rm}%
{\par\smallskip\par}
\newenvironment{subproof}{\par\smallskip\par\noindent{\sl Proof of Claim~\theclaim.~}}%
{$\,\triangleleft$\par\medskip\par}

\makeatletter
\def\@gifnextchar#1#2#3{\let\@tempe#1\def\@tempa{#2}\def\@tempb{#3}%
  \futurelet\@tempc\@gifnch}

\def\@gifnch{\ifx\@tempc\@sptoken\let\@tempd\@tempb%
  \else\ifx\@tempc\@tempe\let\@tempd\@tempa\else\let\@tempd\@tempb\fi\fi\@tempd}

\def\SK@set#1{\left\{#1\right\}}
\def\SK@@set#1#2{\{#1\,:\,
    \begin{array}{@{}l@{}}#2\end{array}
\}}
\def\SK@mset#1{\left\{\!\!\left\{#1\right\}\!\!\right\}}
\def\SK@@mset#1#2{\{\!\!\{#1\,:\,
    \begin{array}{@{}l@{}}#2\end{array}
\}\!\!\}}
\def\BIG@set#1{\Big\{#1\Big\}}
\def\BIG@@set#1#2{\Big\{#1\:\Big|\:
    \begin{array}{@{}l@{}}#2\end{array}
\Big\}}
\newcommand{\Set}[1]{\@gifnextchar\bgroup{\SK@@set{#1}}{\SK@set{#1}}}
\newcommand{\Mset}[1]{\@gifnextchar\bgroup{\SK@@mset{#1}}{\SK@mset{#1}}}
\newcommand{\Bigset}[1]{\@gifnextchar\bgroup{\BIG@@set{#1}}{\BIG@set{#1}}}
\makeatother

\newcommand{\refeq}[1]{(\ref{eq:#1})}
\newcommand{\of}[1]{\left( #1 \right)}
\newcommand{\map}[2]{:#1 \rightarrow #2}

\newcommand{\bZ}{\mathbb{Z}}

\newcommand{\cA}{\mathcal{A}}
\newcommand{\cC}{\mathcal{C}}
\newcommand{\cV}{\mathcal{V}}
\newcommand{\cG}{\mathcal{G}}
\newcommand{\cD}{\mathcal{D}}

\newcommand{\cL}{\mathcal{L}}

\newcommand{\first}{\smallskip\textit{1st round.}~}
\newcommand{\second}{\smallskip\textit{2nd round.}~}
\newcommand{\third}{\smallskip\textit{3rd round.}~}
\newcommand{\fourth}{\smallskip\textit{4th round.}~}
\newcommand{\fifth}{\smallskip\textit{5th round.}~}

\newcommand{\round}[1]{\smallskip\textit{#1th round.}~}
\newcommand{\rounds}[1]{\smallskip\textit{Rounds~#1.}~}

\newcommand{\yes}{\texttt{yes}\xspace}
\newcommand{\no}{\texttt{no}\xspace}

\newcommand{\failed}{\texttt{failed}\xspace}
\newcommand{\ambitious}{\texttt{ambitious}\xspace}
\newcommand{\claiming}{\texttt{claim}\xspace}
\newcommand{\concerned}{\texttt{concerned}\xspace}
\newcommand{\alarm}{\texttt{alarm}\xspace}

\newcommand{\success}{\texttt{success}\xspace}

\newcommand{\m}{\ensuremath{\mathsf{M}}\xspace}
\newcommand{\sbmodel}{\ensuremath{\mathsf{SB}}\xspace}
\newcommand{\mbmodel}{\ensuremath{\mathsf{MB}}\xspace}
\newcommand{\vbmodel}{\ensuremath{\mathsf{VB}}\xspace}
\newcommand{\vvmodel}{\ensuremath{\mathsf{VV}}\xspace}
\newcommand{\mvmodel}{\ensuremath{\mathsf{MV}}\xspace}
\newcommand{\svmodel}{\ensuremath{\mathsf{SV}}\xspace}

\newcommand{\mb}{\ensuremath{\mathsf{MB}}\xspace}
\newcommand{\pn}{\ensuremath{\mathsf{PN}}\xspace}
\newcommand{\local}{\ensuremath{\mathsf{LOCAL}}\xspace}
\newcommand{\congest}{\ensuremath{\mathsf{CONGEST}}\xspace}
\newcommand{\bcongest}{\ensuremath{\mathsf{Broadcast\ CONGEST}}\xspace}
\newcommand{\bcong}{\ensuremath{\mathsf{B\mbox{-}CONGEST}}\xspace}
\newcommand{\all}{\ensuremath{\mathsf{ALL}}\xspace}

\newcommand{\probl}[1]{\textsc{#1}\xspace}
\newcommand{\umaxdeg}{\probl{Unique MaxDegree}}
\newcommand{\leader}{\probl{Leader Election}}
\newcommand{\tridetect}{\probl{Triangle Detection}}
\newcommand{\trifind}{\probl{Triangle Finding}}
\newcommand{\hamilton}{\probl{Hamiltonicity}}

\newcommand{\gnp}{^{G(n,p)}}
\newcommand{\aas}{_{a.a.s.}}
\newcommand{\sa}{_\mathrm{snd}}
\newcommand{\nsa}{_{n\text{-}\mathrm{snd}}}

\newcommand{\short}{^*}

\DeclareMathOperator{\id}{\mathrm{id}}
\DeclareMathOperator{\diam}{\mathrm{diam}}
\DeclareMathOperator{\ecc}{\mathrm{ecc}}

\DeclareMathOperator{\E}{\mathbb{E}}

\newcommand{\vincl}{\rotatebox[origin=c]{270}{$\subseteq$}}
\newcommand{\veq}{\rotatebox[origin=c]{90}{$=$}}

\newcommand{\prob}[1]{\mathbb{P}( #1 )}

\title{What can be computed in average anonymous networks?}

\author{Joel Rybicki\thanks{Institut f\"ur Informatik,
  Humboldt-Universit\"at zu Berlin, Unter den Linden 6, D-10099 Berlin.}\qquad
  Oleg Verbitsky\thanks{Institut f\"ur Informatik,
    Humboldt-Universit\"at zu Berlin, Unter den Linden 6, D-10099 Berlin.
    Funded by the Deutsche Forschungsgemeinschaft (DFG, German Research Foundation) – project number 572124308.
  On leave from the IAPMM, Lviv, Ukraine.}\qquad
Maksim Zhukovskii\thanks{The University of Sheffield, School of Computer Science, Sheffield S1 4DP, UK.}}

\date{}

\begin{document}

\maketitle

\begin{abstract}
We study what \emph{deterministic distributed algorithms} can compute on \emph{random input graphs} in extremely weak models of distributed computing: all nodes are anonymous, and in each communication round, nodes broadcast a message to all their neighbors, receive a (multi)set of messages from their neighbors, and update their local state. These correspond to the \sbmodel{} and \mbmodel{} models introduced by Hella et al. [PODC 2012] and are strictly weaker than the standard port-numbering \pn{} and \local{} models.

﻿We investigate what can be computed almost surely on random input graphs.
We give a \emph{one-round} deterministic $\sbmodel$-algorithm using $O(\log n)$-bit messages that computes unique identifiers with high probability on anonymous networks sampled from $G(n,p)$, where $n^{\varepsilon-1} \le p \le 1/2$   and $\varepsilon>0$ is an arbitrarily small constant.
This algorithm is inspired by canonical labeling techniques in graph isomorphism testing and can be used to ``anonymize'' existing distributed graph algorithms designed for the broadcast \congest{} and \local{} models.
In particular, we give a new anonymous algorithm that finds
a triangle in $O(1/\varepsilon)$ rounds on the above input distribution.

﻿We also investigate computational power of natural analogs of ``Monte Carlo'' and ``Las Vegas'' distributed graph algorithms in the random graph setting, and establish some new collapse and hierarchy results. For example, our work shows the collapse of the weak model hierarchy of Hella et al. on $G(n,p)$, as apart from a vanishingly small fraction of input graphs, the \sbmodel{} model is as powerful as \local{}.
\end{abstract}

\section{Introduction}\label{s:intro}

Many graph problems cannot be solved by deterministic, anonymous distributed algorithms for the simple reason that symmetry cannot be broken in worst-case graphs~\cite{Angluin80,yamashita1988computing,yamashita2002computing}. Thus, the typical approach is to assume nodes are given some symmetry breaking information as input, such as unique identifiers (e.g., the \local{} model~\cite{Linial1992}), a stream of random bits (e.g., the port-numbering model \pn{} with private random bits), or locally unique identifiers (e.g., a suitable graph coloring~\cite{emek2014anonymous}).

But if we move beyond \emph{worst-case} inputs, to which extent is such additional symmetry breaking information necessary?
Can anonymous, deterministic distributed algorithms perform well on \emph{most} graphs?
To address these questions,  we study what can be computed by anonymous, deterministic algorithms on
random input graphs sampled from the
Erd\H{o}s-R\'enyi random graph distribution $G(n,p)$.
So far, even in the non-anonymous setting, the average-case complexity of distributed graph algorithms has remained largely unexplored apart from few isolated examples~\cite{levy2004distributed,krzywdzinski2015distributed,chatterjee2018fast,Turau20}.

\paragraph{Our results in a nutshell.}
We investigate the question of what can be deterministically computed on ``most'' anonymous networks. We introduce and study the notions of asymptotically almost surely correct (``Monte Carlo'') and sound (``Las Vegas'') deterministic distributed algorithms on random input graphs.
It turns out
-- somewhat surprisingly --
that apart from a vanishingly small fraction of graphs, anonymous deterministic algorithms in \emph{very} weak models of computing can solve any problem solvable in the \local{} model
with only a small overhead in time and message length.

This is because on most graphs identifiers can be deterministically generated in only \emph{one round} using short $O(\log n)$-bit messages, allowing simulation of both \local{} and (broadcast) \congest{} algorithms in weak anonymous models of computing.
This further implies
that a range of non-trivial tasks, such as finding Hamiltonian cycles or triangles, can be done fast on anonymous random~graphs.

\subsection{The setting: weak models of distributed computing}

We consider deterministic and anonymous distributed graph algorithms in the \emph{set-broadcast} $\sbmodel$ and the \emph{multiset-broadcast} $\mb$ models. These models lie at the bottom of the weak model hierarchy introduced by Hella et al.~\cite{HellaJKLLLSV15}. In particular, both $\sbmodel$ and $\mbmodel$ models are strictly weaker than the standard  port-numbering \pn{} and $\local$ models of distributed computing.

In the distributed setting, the input graph $G = (V,E)$ on $n$ nodes describes the communication topology of the distributed system, where nodes represent processors and edges communication links between processors.
Each node $v \in V$ is anonymous, but may additionally receive some local input $f(v)$, depending on the problem and the model of computing.
In the \emph{set-broadcast model} $\sbmodel$, each node $v \in V$ receives as local input its degree $\deg(v)$.
All nodes run the same algorithm and computation proceeds in synchronous rounds, where in each round every node $v \in V$ in parallel
\begin{enumerate}[noitemsep]
  \item \emph{broadcasts} a single message to all its neighbors,
  \item receives a \emph{set} of messages broadcast by its neighbors,~and
  \item updates its local state.
\end{enumerate}
Once a node enters a stopping state, it halts with some output $g(v)$.
In the \sbmodel{} model, nodes cannot distinguish which neighbor sent which incoming message nor count the multiplicities of messages.
In the \emph{multiset-broadcast model} $\mb$, nodes receive \emph{multisets} of messages instead of sets. That is, a node can also count the \emph{multiplicity} of received messages. Otherwise, the two models are the same.\footnote{%
In particular, as in $\sbmodel$, we assume that in $\mbmodel$ nodes know their degrees initially, although this could be learned in one communication round.}


The standard \local{} model is equivalent to the setting in which each node is given a unique identifier as local input. We also consider the ``congested'' variants of these models, where in each round nodes broadcast only short $O(\log n)$-bit messages. While the congested version of \local{} is known as the \congest{} model, we denote the congested versions of $\sbmodel$ and $\mbmodel$ by $\sbmodel\short$ and $\mbmodel\short$, respectively. Note that the $\sbmodel\short$ model with unique identifiers is (up to constant factors) equivalent to the \emph{broadcast} \congest{} model (abbreviated as \bcong{}), where in each round each node broadcasts a single $O(\log n)$-bit message to all of its neighbors.

\subsection{Distributed graph problems}\label{ss:definitions}

\paragraph{Graph problems.}
We consider problems on unlabeled graphs (i.e., the nodes or edges of the graph do not have input labels).
Let $\Gamma$ be a (possibly infinite) set of output labels.
A \emph{graph problem} $\Pi$ associates to each graph $G = (V,E)$ a set $\Pi(G) \subseteq \Gamma^V$ of solutions.
That is, a labeling $\sigma \colon V \to \Gamma$ is a solution to $\Pi$ on input $G$ if $\sigma \in \Pi(G)$.
We say that a problem $\Pi$ is \emph{isomorphism-invariant} if for any isomorphic pair of graphs $G$ and $H$ with an isomorphism $\alpha$ from $G$ to $H$ we have that
\[
\Pi(G) = \{ \sigma \circ \alpha : \sigma \in \Pi(H)\}.
\]
 We say that $\Pi$ is a \emph{component-wise} problem if for every graph $G$ and labeling $\sigma \colon V \to \Gamma$ we have that $\sigma \in \Pi(G)$ if and only if $\sigma_{\restriction U} \in \Pi(H)$ for every connected component $H=(U,F)$ of $G$,
where $\sigma_{\restriction U} \colon U \to \Gamma$ is the restriction of $\sigma$ to the set $U \subseteq V$.
 Throughout, we consider only isomorphism-invariant, component-wise graph problems. This definition covers many typical problems,
 in particular, all problems on unlabeled graphs that are solvable in the \emph{uniform} \local{} model (where nodes do not a priori know~$n$).

\paragraph{Solvability and time complexity.}
For any model $\m$, we say that an algorithm $\cA$ \emph{solves} a graph problem $\Pi$ if on every input $G = (V,E)$, each node $v \in V$ halts with an output $\sigma(v)$ and $\sigma \in \Pi(G)$.
We often write $\cA(G)$ for the labeling $\sigma$ output by the algorithm $\cA$ on the input graph~$G$.

The \emph{time complexity} $T_\cA(G)$ of an algorithm $\cA$ on the input graph $G$ is the number of steps until all nodes have reached a stopping state.
We say an algorithm is \emph{non-uniform} if each node receives as local input also the value $n = |V|$. Otherwise, it is \emph{uniform}.
Unless specified otherwise, all our algorithms are uniform.

For a given model $\m$, we use $\m$ to denote the class of problems solvable in the model.
We write $\m[r]$ to denote the class of problems solvable in $r$ rounds in this model.
We write $\all$ for the set of all isomorphism-invariant, component-wise problems.
\Cref{table:glossary} in Appendix~\ref{apx:glossary}
compares and summarizes the different models we consider in this work.

\subsection{Average-case distributed computation}\label{s:average}

We now introduce the formal framework we use to study anonymous distributed graph algorithms in the average case. We consider  analogs of Monte Carlo (almost sure solvability) and Las Vegas (soundness) concepts for deterministic algorithms on random input graphs.

\paragraph{Random graphs.}
An \emph{Erd\H{o}s-R\'enyi random graph $G(n,p)$} is an $n$-node graph,
where each pair of nodes is adjacent with probability $p=p(n)$, independently of all other pairs.
In particular, $G(n,1/2)$ is a random graph chosen equiprobably from among all graphs on the node set~$V = \{1,\ldots,n\}$.
We say that an event happens for $G(n,p)$ \emph{asymptotically almost surely}
(\emph{a.a.s.} for brevity) if the probability of this event tends to 1 as $n\to\infty$, i.e., with probability $1-o(1)$.
We write $G \sim G(n,p)$ to emphasize that the graph $G$ is sampled from the  Erd\H{o}s-R\'enyi distribution.

\paragraph{Asymptotically almost sure computation.}
  Let $\mathcal{D}(n)$ be a probability distribution over $n$-node graphs with $V = \{1, \ldots, n\}$.
  An algorithm $\cA$ solves a problem $\Pi$ \emph{asymptotically almost surely (a.a.s.)} on the distribution $\mathcal{D}(n)$ if for the random graph $G \sim \mathcal{D}(n)$ we have that
  \[
  \mathbb{P}(\cA(G) \in \Pi(G)) = 1 - o(1).
  \]
  When we do not specify the distribution, we assume that the distribution is given by $G(n,1/2)$, i.e., the uniform distribution over all $n$-node graphs. In this case, we often simply say that $\cA$ solves the problem asymptotically almost surely. This means $\cA$ outputs a correct solution to $\Pi$ on a fraction of the $n$-node graphs that approaches 1 as $n \to \infty$, that is, on \emph{almost all} $n$-node graphs.

  The class of problems that are solvable a.a.s.\ on the distribution $\mathcal{D}(n)$ by an \m algorithm is denoted by~$\m\aas^{\mathcal{D}(n)}$.
  Similarly, we write $\m\aas^{\mathcal{D}(n)}[r]$ for the class of problems solvable on $\mathcal{D}(n)$ a.a.s.\ in at most $r$ rounds by an \m algorithm.
  As before, we omit for brevity the distribution $\mathcal{D}(n)$ from the notation when $\mathcal{D}(n) = G(n,1/2)$.
  Note that algorithms that solve a problem a.a.s.\ are analogous to \emph{Monte Carlo algorithms}, but the randomness is over the input graphs rather than the random choices of the algorithm (which deterministic algorithms cannot make).

\paragraph{Soundness.}
Let $\Gamma$ be a set of output labels and $\failed$ be a special label not in $\Gamma$.
We say that an algorithm $\cA$ for a problem $\Pi$ is \emph{sound} if
\begin{enumerate}[noitemsep]
  \item $\cA$ solves $\Pi$ asymptotically almost surely, \emph{and}
\item  if $\cA(G) \notin \Pi(G)$, then $\cA(G)$ is the constant labeling $\sigma(\cdot) = \failed$.
\end{enumerate}
That is, if the algorithm fails to output a solution, which happens only with probability $o(1)$, then all nodes output the special value \failed to indicate this. We write $\m\sa$ to denote the class of distributed problems that admit a sound algorithm in the model~\m.

Note that sound algorithms can be seen analogous to \emph{Las Vegas algorithms}, but the randomness is over the input graphs rather than the random choices of the algorithm: such algorithms always terminate, but they can terminate with the knowledge that no solution was found. In other words, a sound algorithm cannot output an incorrect solution without all nodes knowing this.

\paragraph{Expected time complexity.}
A problem $\Pi$ is solvable in \m in \emph{expected time} $T(n)$ on the distribution $\mathcal{D}(n)$ if there is an $\m$ algorithm $\cA$ solving $\Pi$ correctly on every input graph such that
$\E_{G \sim \mathcal{D}(n)}( T_\cA(G) ) \le T(n)$ for all~$n$.

\subsection{Our contributions}

\subsubsection{Anonymization of \local{} and \bcongest{} algorithms}
Our first result shows that unique identifiers can be generated asymptotically almost surely in one communication round.
The only communication needed is that each node broadcasts a single $O(\log n)$-bit message, namely its degree.
The algorithm
is based on a new succinct \emph{graph canonical labeling}.
Such labelings have been widely studied in the context of graph isomorphism testing; see  \Cref{table:summary} for a summary and \Cref{s:related} for further discussion.
Unlike the existing labeling schemes, our new approach admits a communication-efficient implementation that produces short identifiers in the weak $\sbmodel\short$ model of distributed computing (i.e., set-broadcast with $O(\log n)$-bit messages).

\begin{restatable}{theorem}{congestpp}
\label{thm:congest-pp}
Let $\varepsilon>0$ be a constant.
Then there exists a one-round $\sbmodel\short$ algorithm (depending only on $\varepsilon$) that
asymptotically almost surely assigns unique $O(\log n)$-bit identifiers on input $G(n,p)$ for any
$n^{\varepsilon-1}\le p\le1/2$.
\end{restatable}

The algorithm  in the above theorem is simple, but its analysis is non-trivial.
This result has several important consequences for weak models of distributed computing.
In particular, it readily implies the following inclusions:
\begin{eqnarray}
  \bcong^{G(n.p)}\aas[r]&\subseteq&\sbmodel^{*,\,G(n.p)}\aas[r+1]\text{ and}\label{eq:simcong}\\[2mm]
  \local^{G(n.p)}\aas[r]&\subseteq&\sbmodel^{G(n.p)}\aas[r+1]\label{eq:simloc}
\end{eqnarray}
for any $n^{\varepsilon-1}\le p\le1/2$ and $r=r(n)$.

We can use the first of these inclusions to immediately ``anonymize'' algorithms for random graphs.
For example, Turau~\cite{Turau20} gave a fast \bcongest-algorithm for finding Hamiltonian cycles on $G(n,p)$. Applying \Cref{thm:congest-pp} to Turau's algorithm immediately yields the following result.

\begin{restatable}{corollary}{hcalgo}\label{cor:hcalgo}
  There exists an $\sbmodel\short$ algorithm which
  asymptotically almost surely finds a Hamiltonian cycle in $G(n,p)$ within $O(\log n)$ rounds
  for any $p\ge\log^{3/2}n/\sqrt n$.
\end{restatable}

Naturally, we can also anonymize not just algorithms on random graphs. For example, Korhonen and Rybicki~\cite{korhonen2017broadcast} gave deterministic algorithms in the broadcast \congest{} model to find trees of size $k$ in $O(k2^k)$ rounds and even-length cycles in $O(n)$ rounds. By \Cref{thm:congest-pp}, these algorithms can be translated to anonymous $\sbmodel\short$-algorithms that a.a.s.\ find such subgraphs asymptotically as fast.

\begin{table}[t]
\begin{center}
\renewcommand{\arraystretch}{1.6}
\begin{tabular}{@{}l@{\extracolsep{5mm}}l@{\extracolsep{5mm}}l@{\extracolsep{5mm}}l}
\toprule

\textbf{Edge probability}  & \textbf{Model} & \textbf{ID length}  & \textbf{Reference}\\
\hline%
  $p=\frac12$  & $\sbmodel\short[1]$ & $\Theta(\log^2 n)$ & Babai, Erd\H{o}s, and Selkow~\cite{BabaiES80} \\
  $n^{-1/5}\ln n\ll p\leq 1/2$ & $\sbmodel\short[1]$ & $\Theta(\frac{\log n\log\Delta}p)$ &  Bollob\'{a}s~\cite[Ch.~3.5]{Bollobas_book} \\
$\frac{\ln^2 n}{n}\ll p\leq \frac{1}{2}$ & $\mb\short[1]$ & $\Delta\log\Delta$ &  Czajka and Pandurangan~\cite{CzP}, \\
&&& Mossel and Ross~\cite{MosselR19}\\
  $\frac{(1+\delta)\ln n}{n}\le p\ll n^{-5/6}$ & $\mb\short[2]$ & $\Theta(\Delta\log\Delta)$  & Gaudio, R\'{a}cz, and Sridhar \cite{GaudioRS25}\\
  $p=\frac12$ & $\mbmodel\short[2]$ & $O(\log n)$ & Verbitsky and Zhukovskii~\cite{VerbitskyZ23esa} \\
  $n^{\varepsilon-1}\le p\leq 1/2$  & $\sbmodel\short[1]$ & $O(\log n)$ & \textbf{New: This work} \\
\bottomrule
\end{tabular}
\end{center}
\caption{Unique identifiers based on canonical labeling algorithms for $G(n,p)$. The identifier  lengths are in bits. Here, $\Delta$ denotes the maximum degree of the random graph; a.a.s.{} $\Delta=O(np)$.
While the prior papers do not consider weak models of distributed computing, one can verify that the schemes are implementable in the respective models mentioned in the ``Model'' column using $O(\log n)$-bit messages.  In \Cref{ss:collapse} we discuss these distributed implementations in further detail. The value in the brackets indicates the time complexity in the respective model.
\label{table:summary}}
\end{table}

\subsubsection{Anonymous triangle finding with short messages}
As a further example, we investigate triangle finding in congested, anonymous networks. In  triangle \emph{detection}, at least one node outputs ``yes'' if and only if there is a triangle in the graph. In triangle \emph{finding}, at least one node needs to output a witness for the existence of a triangle. In triangle \emph{enumeration}, for each triangle of the graph, there has to be some node that outputs the triangle.
While all variants are trivial in the \local{} model,  the complexity of triangle detection and finding remains a long-standing open problem in the  \congest{} and related  models~\cite{dolev2012tri,drucker2014power,censor2019algebraic,izumi2017triangle,fischer2018possibilities,abboud2020fooling,censor2021distributed,AssadiS25}.

The triangle detection problem is trivial on $G(n,p)$ if nodes know the value of $p$, as there is a threshold at $p=\Theta(1/n)$ for the existence of a triangle; see, e.g., \cite[Theorem 3.4]{Janson_book}.
On the other hand, Izumi and Le Gall~\cite{izumi2017triangle} showed that triangle \emph{enumeration} requires $\Omega(n^{1/3}/\log n)$ rounds in the (non-broadcast) \congest{} model with inputs from $G(n, 1/2)$.
We show that the ``intermediate problem'' of triangle \emph{finding} can be solved substantially faster, in only diameter time, for a large range of $p$.

\begin{restatable}{theorem}{triangle}\label{thm:triangle}\hfill
  \begin{enumerate}[label=\normalfont \textbf{\arabic*}.]%
  \item
  There exists a sound \bcong algorithm  which, for any constant $\varepsilon>0$ and
  any $p$ satisfying $n^{\varepsilon-1}\le p\le1/2$, asymptotically almost surely
  finds a triangle on $G \sim G(n,p)$ in~$O(\diam(G))$~rounds.
  \item
  There exists a \bcong algorithm  which, for any constant $\varepsilon>0$ and
  any $p$ satisfying $n^{\varepsilon-1}\le p\le1/2$,
  finds a triangle on $G \sim G(n,p)$ in $O(1/\varepsilon)$ expected rounds.
  \end{enumerate}
\end{restatable}

This result is based on a new ``eccentricity lemma'' that identifies a useful combinatorial property for triangles in $G(n,p)$.
Using \Cref{thm:congest-pp}, these algorithms can be translated to the anonymous $\sbmodel\short$; see~\refeq{simcong}.
That is, even when nodes do not have identifiers, nodes can quickly find and identify a triangle in the graph asymptotically almost surely.

\begin{restatable}{corollary}{diametertriangle}\label{cor:diametertriangle}
  Let $\varepsilon>0$.
  There exists an $\sbmodel\short$ algorithm which, for any
  $p$ satisfying $n^{\varepsilon-1}\le p\le1/2$, asymptotically almost surely
  finds a triangle in $G \sim G(n,p)$ in~$O(\diam(G))$~rounds.
\end{restatable}

\subsubsection{Collapse of the weak model hierarchy in the average case}\label{sss:collapse}
It is well-known that the port-numbering model \pn{} is strictly weaker than the \local{} model.
Hella et al.~\cite{HellaJKLLLSV15} established the following strict hierarchy for even weaker models of distributed computing
\[
\sbmodel \subsetneq \mb \subsetneq \vvmodel \subsetneq \pn \subsetneq \local,
\]
where $\vvmodel$ is a variant of the \pn{} model where the incoming and outgoing port are numbered independently.
First, we observe that the weak model hierarchy collapses in the average case; see~\refeq{simloc}.

Second, we show new \emph{time hierarchy collapse results} for weak models.
We show that when $p$ is above the connectivity threshold of $G(n,p)$,
any distributed problem can be solved on $G(n,p)$
asymptotically almost surely by an anonymous \mb{}-algorithm in $o(\log n)$ rounds.
If $p\ge n^{\varepsilon-1}$ for $\varepsilon>0$,
the number of rounds is $O(1)$. Using Theorem \ref{thm:congest-pp}, we
obtain a similar result even for the weak model~\sbmodel.
In both cases, the time hierarchy collapses to the level equal to the expected diameter of $G(n,p)$
up to a small additive constant.
Formally, we prove the following in \Cref{s:collapse}.

\begin{restatable}{theorem}{collapse}\label{thm:collapse}\hfill
  \begin{enumerate}[label=\normalfont \textbf{\arabic*}.]%
  \item
    Let $\frac{(1+\delta)\ln n}{n}\le p\le\frac{1}{2}$ for a constant $\delta>0$. Then
    $$
\all=\mb\gnp\aas\left[(1+o(1))\frac{\ln n}{\ln(np)}+4\right].
$$
  \item
    Let $n^{\varepsilon-1}\le p\le\frac{1}{2}$ for a constant $\varepsilon>0$. Then
    $$
\all=\sbmodel\gnp\aas\left[(1+o(1))\frac{\ln n}{\ln(np)}+3\right]=\sbmodel\gnp\aas\left[(1+o(1))/\varepsilon+3\right].
$$
  \end{enumerate}
\end{restatable}

\subsubsection{Limitations and power of sound algorithms}
Finally, we investigate the difference between asymptotically almost surely correct (``Monte Carlo'') algorithms and sound (``Las Vegas'') algorithms.
One issue with the former is that although a correct solution is found on almost all graphs,
the nodes cannot be certain that this is the case for any particular input graph $G$.
As defined in \Cref{s:average}, sound algorithms output a correct solution a.a.s.; if the output is not a correct solution, then all nodes declare that the algorithm~failed.

It is not hard to see that sound and uniform \local{} algorithms are (strictly) stronger than worst-case \local{} algorithms (which must work on all inputs). This holds even for constant-round algorithms; see Appendix \ref{app:UMD} for an illustrative example. However, in the \mbmodel{} model, the power of sound and uniform algorithms is quite limited: they cannot solve problems such as leader election, triangle detection, or Hamiltonicity recognition; see Section~\ref{ss:triangle-nonsound}. In particular, Corollaries \ref{cor:hcalgo} and \ref{cor:diametertriangle} are optimal in the sense that no sound, uniform \mbmodel{} algorithms exist for Hamiltonicity or triangle detection.

Nevertheless, we show that anonymity is not a barrier for \emph{non-uniform} sound algorithms, i.e., when nodes get the number $n$ of nodes as input.
We write $\m\nsa$ to denote the class of distributed problems solvable in the model \m
by non-uniform sound algorithms. We study the round hierarchy
$$
\m\nsa[1]\subseteq\m\nsa[2]\subseteq\cdots\subseteq\m\nsa[r]\subseteq\m\nsa[r+1]\subseteq\cdots
$$
for the classes $\m\in\{\mb,\sbmodel\}$.
We show that this hierarchy collapses
at the fourth level---for the remarkable reason that all problems are solvable by
non-uniform sound \sbmodel{} algorithms in just four rounds.
Furthermore, we show that the collapse does not occur earlier. In summary, we get that
$$
\begin{array}{ccccccccc}
\sbmodel\nsa[1] & \subseteq & \sbmodel\nsa[2] & \subseteq & \sbmodel\nsa[3] & \subsetneq & \sbmodel\nsa[4] & = & \all \\
\vincl &&   \vincl &&   \vincl &&   \veq &&  \\
\mb\nsa[1] & \subseteq & \mb\nsa[2] & \subseteq & \mb\nsa[3] & \subsetneq & \mb\nsa[4]. &&
\end{array}
$$
This result is formally captured by the following theorem we prove in \Cref{s:collapse}.

\begin{restatable}{theorem}{allnsacollapse}
\label{thm:all-nsa}\hfill
  \begin{enumerate}[label=\normalfont \textbf{\arabic*}.]
  \item
  $\all=\sbmodel\nsa[4]$.
  \item
 $\mb\nsa[3]\subsetneq\all$.
  \end{enumerate}
\end{restatable}

\subsection{Related work}\label{s:related}

\paragraph{Deterministic anonymous algorithms.}
Foundational work in distributed computing has established that many symmetry-breaking problems, such as leader election and graph coloring, cannot be deterministically solved in anonymous networks in the port-numbering model \pn{}~\cite{Angluin80,yamashita1988computing}.
However, various optimization problems on graphs can be \emph{approximated} fast even in the \pn{} model~\cite{aastrand2010fast,polishchuk2009simple,suomela2010distributed}.

Hella et al.~\cite{HellaJKLLLSV15} studied what can be computed in models that are strictly weaker than the port-numbering model $\pn{}$. They showed there exists a  hierarchy of weak models of computing
\[
\sbmodel \subsetneq \mb = \vbmodel \subsetneq \svmodel = \mvmodel = \vvmodel \subsetneq \pn \subsetneq \local,
\]
where in  the model $\vbmodel$ nodes have incoming port-numbers but broadcast the same message to all neighbors, in $\svmodel$ nodes receive a set of messages but have outgoing port numbers (and can send distinct messages to each neighbor), and  $\vvmodel$ is a model, where nodes have both incoming and outgoing port-numbers (that can be inconsistent unlike in \pn{}).

Lempiäinen~\cite{lempiainen2016ability} showed that the cost of simulating $\mvmodel$ in $\svmodel$ is $\Theta(\Delta)$ communication rounds in the worst case. In contrast, our work shows that simulating the \local{} model asymptotically almost surely in $\sbmodel$ requires only \emph{one} additional communication round, and that the same holds for the congested models: broadcast \congest{} can be simulated with the cost of one additional round in~$\sbmodel\short$.

\paragraph{Randomization in anonymous networks.}
When nodes have access to a stream of random bits, any \local{} algorithm can be simulated with high probability in the \pn{} model as nodes can simply generate random unique identifiers.
Thus, in the randomized setting, work has focused on understanding how much randomness is needed~\cite{emek2014anonymous,seidel2015randomness,kowalski2025random}.
Emek et al.~\cite{emek2014anonymous} showed that randomized anonymous (Las Vegas) \pn algorithms for decision problems can be derandomized given a distance-2 coloring of the nodes and Seidel et al.~\cite{seidel2015randomness} investigated how many random bits are needed for the derandomization. Recently, Kowalski et al.~\cite{kowalski2025random} studied how many random bits are needed to solve leader
election in anonymous networks.
In addition to randomized \pn{} model, weaker (anonymous) models of computing have been studied, such as the stone-age model~\cite{emek2013stone} and the beeping model~\cite{cornejo2010deploying,giakkoupis2026self,vacus2025minimalist}.

\paragraph{Distributed algorithms on random graphs.}
Krzywdzi{\'n}ski and Rybarczyk~\cite{krzywdzinski2015distributed} studied
 \local{} algorithms on $G(n,p)$ for various local symmetry-breaking problems, such as maximal independent sets, colorings, matchings, and so on. Levy, Louchard, and Petit~\cite{levy2004distributed}, Chatterjee et al.~\cite{chatterjee2018fast} and Turau~\cite{Turau20} studied the problem of finding a Hamiltonian cycle in the \congest{} model. The state-of-the-art algorithm by Turau runs in $O(\log n)$ rounds for $p \ge (\log n)^{3/2} / \sqrt{n}$.
 Our results can be used to ``anonymize'' existing algorithms for the regime $n^{\varepsilon-1} \le p \le 1/2$.

Another line of research has studied graphs with good \emph{expansion} properties, which allow efficient routing of information using randomized~\cite{ghaffari2017distributed,ghaffari2018new} and deterministic~\cite{chang2024deterministic,chang2020deterministic} \congest{} algorithms. Since random graphs are expanders a.a.s., these results can also be used to obtain algorithms for random graphs as well.
However, the time complexity of these algorithms tends to grow fast with $n$, with running times of order $2^{\Omega(\sqrt{\log n})}$.
Recently, Maus and Ruff~\cite{maus2026distributed} studied fast randomized coloring hyperbolic random graphs in the \congest{} model.

In addition, various computationally-restricted (and often asynchronous) models  have also been investigated on random graphs; see e.g.~\cite{alistarh2025near,giakkoupis2023distributed,panagiotou2015randomized}.
Finally, we note that an orthogonal line of research on \emph{node-average complexity} studies the running time of distributed algorithms of an average node in worst-case inputs rather than on average random inputs~\cite{feuilloley2020long,balliu2022node}.

\paragraph{Distributed triangle detection, finding and enumeration.}
Detecting, finding, and listing triangles are fundamental problems in distributed computing~\cite{censor2021distributed}.
In  \local{}, these problems are solvable in constant time and a trivial \congest{} algorithm solves the problems in $O(n)$ time.

Dolev et al.~\cite{dolev2012tri} gave a deterministic algorithm that runs in $O(n^{1/3} / \log n)$ rounds in the congested clique model. In the same model, Censor-Hillel et al.~\cite{censor2019algebraic} gave triangle detection and counting algorithms that run in $O(n^\rho)$ rounds, where $\rho$ is the distributed matrix multiplication exponent.
Dolev et al.~\cite{dolev2012tri} also gave a sampling-based randomized algorithm that finds a triangle in $O(n^{1/3} / (t^{2/3} + 1))$ rounds with high probability when the input graph has $t$ triangles.
Recently, Censor-Hillel et al.~\cite{censor2024faster} gave a faster sampling algorithm  that runs in $\tilde{O}(n^{0.1567}/(t^{0.3926}+1))$ time.

Note that for $G(n,p)$, the number of triangles for $p = n^{\varepsilon-1}$ is of order $n^{3\varepsilon}$. So for sufficiently large $\varepsilon>0$, the above randomized protocols run in (near) constant-time. However, our new algorithm finds a triangle in constant expected time for any constant $\varepsilon>0$ and works in the anonymous $\sbmodel\short$ rather than the congested clique model.

Lower bounds for triangle detection and finding are notoriously challenging, as super-logarithmic lower bounds for triangle detection imply breakthroughs in circuit complexity~\cite{drucker2014power,eden2022sublinear}.
There has been substantial prior work on lower bounds for triangle detection in distributed models; see, e.g.,~\cite{abboud2020fooling,fischer2018possibilities}. 
In a recent breakthrough, Assadi and Sundaresan~\cite{AssadiS25} showed that $\Omega(\log \log n)$ rounds are needed for triangle \emph{detection} in the \congest{} model using a round elimination argument.

Izumi and Le Gall~\cite{izumi2017triangle} studied triangle finding
and triangle enumeration.
They gave the first sublinear time algorithms for the general \congest{} model: a $O((n \log n)^{2/3})$-time algorithm for finding triangles and a $O(n^{3/4} \log n)$-time algorithm for triangle enumeration.
They also gave a lower bound of $\Omega(n^{1/3}/\log n)$ rounds for triangle enumeration in \congest{} and congested clique. This lower bound uses an information theoretic argument on uniform random graphs.
Censor-Hillel et al.~\cite{CensorHillelLV24} gave a near-optimal \congest{} algorithm that runs in $n^{1/3 + o(1)}$ rounds.
The lower bound was subsequently matched up to polylogarithmic factors by the algorithm of Chang et al.~\cite{chang2021near}, which solves triangle enumeration in $\tilde{O}(n^{1/3})$ rounds using expander decompositions and expander~routing~\cite{chang2024deterministic,chang2020deterministic}.

\paragraph{Canonical labelings of graphs.}
An important role in isomorphism testing is played by \emph{color refinement} and related algorithms: they can be used to compute canonical
(i.e., isomorphism-invariant) labelings of graphs.
Interestingly, many of these canonization schemes admit distributed implementations, and can be used to compute unique identifiers on various  graph distributions.

\Cref{table:summary} summarizes prior work from the perspective of weak models of distributed computing.
Babai, Erd\H{o}s, and Selkow~\cite{BabaiES80} showed that a rather succinct version of color refinement a.a.s.{}
finds unique labels for nodes in a uniform random graph $G(n,1/2)$; see Section \ref{ss:CR} for details.
Bollob\'{a}s~\cite[Ch.~3.5]{Bollobas_book} proved that
their approach works a.a.s.{} for $G(n,p)$ whenever $p\gg n^{-1/5}\ln n$.
Czajka and Pandurangan~\cite{CzP} established that the color refinement algorithm computes a.a.s.{} canonical labeling of $G(n,p)$ in the sparse setting, when  $p\gg\frac{\ln^4 n}{n\ln\ln n}$, which was later improved to $p\gg\frac{\ln^2 n}{n}$ by Mossel and Ross~\cite{MosselR19}.
Recently, Gaudio, R\'{a}cz, and Sridhar \cite{GaudioRS25} suggested an algorithm that succeeds on $G(n,p)$ for $p\ge\frac{(1+\delta)\ln n}{n}$, for any constant $\delta>0$.

Unfortunately, when translating these labeling schemes to the distributed setting, they incur a high cost in communication complexity, as the resulting identifiers are large, i.e., use super-logarithmically many bits.
As most distributed algorithms rely on communicating identifiers as part of their messages, the long identifiers can inflict in the worst-case a \emph{polynomial slow-down} in the simulation of \congest{} algorithms.
This is because the maximum degree $\Delta$ of a random graph from $G(n,p)$ is a.a.s.{}
$\Delta=(1+o(1))\,np$ if $p\gg \ln n/n$.

However, not all canonical labeling schemes suffer from a high communication complexity in the distributed setting.
Verbitsky and Zhukovskii~\cite{VerbitskyZ23esa} gave a succinct labeling scheme that generates identifiers and uses messages of at most logarithmic length, but their results only apply in uniform random graphs.

In contrast to prior work, our new labeling scheme uses messages and generates identifiers with only logarithmic length, and works for a much larger range of $p$. Moreover, the scheme is time-optimal, as it uses only one communication round.
Our work also answers the question recently raised by Gaudio et al.~\cite{GaudioRS25}: ``can the assumptions on node identifiability be relaxed -- or even removed -- for local algorithms on random graphs''? Our work shows that this is possible even in far weaker models of distributed computing than the \local{} model, and even with short messages.

\subsection{Structure of the paper}

\Cref{s:ids} describes and analyzes our new one-round $\sbmodel\short$ algorithm that assigns small unique identifiers in random graphs.
\Cref{s:triangle} gives a new fast anonymous algorithm that finds a triangle in a random graph.
\Cref{s:collapse} establishes time hierarchy collapse for weak models of computing in the average case.
Finally, in \Cref{s:conclusions} we discuss some open problems left by our work.

\section{Generating unique identifiers with logarithmic communication}\label{s:ids}

In this section, we show that for a large range of $p$, if a distributed problem can be solved on $G(n,p)$ asymptotically almost surely in \bcong{}, then it can also be solved in $\sbmodel\short$ at the cost of one additional communication round. This may come as some surprise, since \sbmodel{} is an extremely weak model. Nevertheless, it turns out that one round in $\sbmodel\short$ suffices for anonymous processors in the $G(n,p)$ network to generate unique identifiers.

\congestpp*

The proof occupies the next two subsections; Subsection \ref{ss:first-application} then illustrates how the theorem can be used to anonymize \bcong{} algorithms.

\subsection{One-round $\sbmodel\short$ algorithm}

We write $N(v)$ for the set of neighbors of a node $v$.
Furthermore, let
$$
S(v)=\Set{\deg u}_{u\in N(v)}
$$
be the set of all distinct degrees of the neighbors of~$v$.

We begin the proof of Theorem \ref{thm:congest-pp} with the algorithm description.
The algorithm depends on an integer parameter $C>0$, whose value is fixed in advanced.

\newsavebox{\abox}
\savebox{\abox}{$\mathcal{A}_C$}

\begin{algorithm}\label{alg}
\SetAlgoRefName{\usebox{\abox}} 
\caption{Generation of unique identifiers}
\begin{itemize}
\item
Each node $v$ sends its degree $\deg v$ to all its neighbors.
\item
Having received the set $S(v)$, the node $v$ outputs the vector
$$s(v)=(s_0,\ldots,s_{C-1})$$
where $s_r$ is the sum of the numbers in $S(v)$ that are congruent to $r$ modulo~$C$.  
\end{itemize}
\vspace{-1em}
\end{algorithm}

To prove Theorem \ref{thm:congest-pp}, it suffices to show that
for any constant $\varepsilon>0$ there exists $C$ such that,
when the algorithm $\cA_C$ is run on $G=G(n,p)$ with $p$ satisfying $n^{\varepsilon-1}\le p\le1/2$,
then asymptotically almost surely every node computes a unique identifier.
This is ensured by the following lemma.

\begin{lemma}\label{lem:congest-pp}
 For any constant $\varepsilon>0$ there exists $C$ such that if $n^{\varepsilon-1}\le p\le1/2$,
 then a.a.s.{} $s(v)\ne s(v')$ for every two nodes $v\ne v'$ of~$G(n,p)$.
\end{lemma}

It follows that unique identifiers can be assigned to the nodes of $G(n,p)$ in $\sbmodel\short\aas$
just in one round. This holds for a wide range of $p$ approaching the connectivity threshold,
as $\varepsilon>0$ can be chosen arbitrarily small.
Note that this range cannot be made much closer to the connectivity threshold: as a consequence
of \cite[Theorem 1.1]{GaudioRS25}, if $p\le\frac{\ln^2n}{100\,n(\ln\ln n)^3}$, then no \mb algorithm
is able, in two rounds, to generate pairwise distinct node labels a.a.s.{} (even without length restriction).
This follows from \cite[Theorem 1.1]{GaudioRS25} by part 1 of Lemma~\ref{lem:cr-states} shown in \Cref{s:collapse}.

We prove \Cref{lem:congest-pp}, and hence Theorem \ref{thm:congest-pp}, in the next subsection.

\subsection{Proof of Lemma \ref{lem:congest-pp}}\label{ss:proof-pp}

Fix two nodes $u$ and $v$ and, according to the distribution $G(n,p)$, expose first all
edges emanating from $u$ and $v$. Let $X$ denote the set of nodes adjacent neither to $u$
nor to $v$, and let $U=N(u)\setminus N(v)$. Set $n'=|X|$ and $m=|U|$. By the Chernoff bound,
the event
$$
n'=(1+o(1))(1-(2p-p^2))\,n\text{ and }m=(1+o(1))(p-p^2)\,n
$$
holds with probability $1-o(n^{-2})$ and, hence, with probability $1-o(1)$
for every pair $u,v$. In what follows, we therefore can assume that this event holds deterministically.

Next, we expose all edges emanating from $N(v)$ and all edges inside $N(u)$.
Note that while the vector $s(v)$ is already determined, the vector $s(u)$ is not yet determined
as only the degrees of the common neighbors in the set $N(u)\cap N(v)$ are exposed.
Indeed, edges between $U$ and $X$ have not yet been exposed and are to be exposed
in the final step of generating $G(n,p)$.
This step determines a set $\mathbf{S}_X$ which will play a crucial role in our analysis.
To define $\mathbf{S}_X$, we first consider the set
$$
\mathbf{S}_X^0=\{|N_X(x)|\}_{x\in U},
$$
where $N_X(x)=N(x)\cap X$. The number $|N_X(x)|$ determines the degree of $x\in U$
in the completely exposed $G(n,p)$. If $\deg x = \deg y$ for some $y\in N(u)\cap N(v)$,
then the node $x$ does not actually affect the vector $s(u)$. We, therefore, remove
the numbers $|N_X(x)|$ for all such $x$ from $\mathbf{S}_X^0$, and denote the resulting
set by~$\mathbf{S}_X$.

Let $\mathbf{s}_X\in\mathbb{Z}^C$ be the vector of sums of elements
from $\mathbf{S}_X$ with the same residue modulo $C$. The event $s(u)=s(v)$ in $G(n,p)$ identifies a
deterministic vector of integers $s$ such that $\mathbf{s}_X=s$. We shall prove that, for each $s$, 
\begin{equation}
  \label{eq:sXs}
\mathbb{P}(\mathbf{s}_X=s)=o(n^{-2}).  
\end{equation}
This will yield the lemma by applying the union bound over all pairs of nodes~$(u,v)$.

Before proving Bound \refeq{sXs}, we make some preparations.
Without loss of generality, let $U=\{1,\ldots,m\}$. Consider $m$ random variables
\begin{equation}
  \label{eq:xi}
\xi_i=|N_X(i)|.  
\end{equation}
Note that $\xi_i\sim\mathrm{Bin}(n',p)$ and that $\xi_1,\ldots,\xi_m$ are independent.
The multiset
$$
\bm{\cD}=\Mset{\xi_1,\ldots,\xi_m}
$$
plays an important role in the sequel, and we need to introduce some notation for multisets of integers.

Let $\cD=\Mset{d_1,\ldots,d_m}$ be a multiset consisting of $m$, not necessarily distinct integers.
For $r=0,1,\ldots,C-1$, let
$$
\cD_r=\Mset{d\in\cD}{d\equiv r\pmod C}.
$$
Denote the size of $\cD_r$ by $m(r)$ and let $d^r_1\le d^r_2\le\dots\le d^r_{m(r)}$ be the list of all
elements of $\cD_r$ in the non-descending order.
Let $\mu(r)$ be the multiplicity of the largest element $d^r_{m(r)}$ in $\cD_r$, so we have
$$
d^r_1\le d^r_2\le\dots\le d^r_{m(r)-\mu(r)} < d^r_{m(r)-\mu(r)+1} = \cdots = d^r_{m(r)}.
$$
We will need to refer to the vectors in $\bZ^C$ consisting of the two largest elements of the multisets
$\cD_0,\cD_1,\ldots,\cD_{C-1}$. Specifically, we set
$$
a=(d^0_{m(0)-\mu(0)},d^1_{m(1)-\mu(1)},\ldots,d^{C-1}_{m(C-1)-\mu(C-1)})\text{ and }
b=(d^0_{m(0)},d^1_{m(1)},\ldots,d^{C-1}_{m(C-1)}).
$$
In particular, for $r=0,1,\ldots,C-1$,
$$
a_r=d^r_{m(r)-\mu(r)}\text{ and }b_r=d^r_{m(r)}.
$$
Finally, we write $\cD^\circ$ to denote the multiset obtained from $\cD$ by removing all occurrences of the
largest elements $d^0_{m(0)},d^1_{m(1)},\ldots,d^{C-1}_{m(C-1)}$.

We will apply this notation to the random multiset $\bm{\cD}$, writing the respective numbers and vectors in bold.

Define the interval of reals
$$
 J=\biggl[n'p+\sqrt{(2-\varepsilon)n'p(1-p)\ln m},\,n'p+\sqrt{C n'p(1-p)\ln m}\biggr].
$$

\begin{claim}
  Let $\mathcal{E}$ denote the event that, for every $r=0,1,\ldots,C-1$,
$$
 \mathbf{b}_r\in J\text{ and }{\bm\mu}(r)\leq C.
 $$
If $C>22/(\varepsilon(1-\varepsilon))$, then $\prob{\mathcal{E}}=1-o(n^{-10})$.
\label{cl:max-degrees-concentrate}
\end{claim}

\begin{subproof}
  Fix $r\in\{0,1,\ldots,C-1\}$. Let $J(r)$ be the set of all integers in $J$ congruent $r$ modulo $C$,
  and note that $\left|J(r)\cap[0,n'p+\sqrt{2n'p(1-p)\ln m}]\right|=\Theta\left(\sqrt{n'p\ln m}\right)$.
  Consider $\xi\sim\mathrm{Bin}(n',p)$. For
  $d=n'p+\sqrt{2\gamma n'p(1-p)\ln m}$ in $J(r)$, by the local limit theorem we have
$$
 \mathbb{P}(\xi=d)=\frac{1+o(1)}{\sqrt{2\pi n'p(1-p)}}\,m^{-\gamma}.
$$
Therefore, the probability that there exists $x\in U$ and $d\in J(r)$ such that $|N_X(x)|=d$ equals
\begin{align*}
 1-\left(1-\sum_{d\in J(r)}\frac{1+o(1)}{\sqrt{2\pi n'p(1-p)}}\,m^{-\gamma}\right)^m&\geq
 1-\exp\left(-\sum_{d\in J(r)}\frac{1+o(1)}{\sqrt{2\pi n'p(1-p)}}\,m^{1-\gamma}\right)\\
 &=
 1-\exp\left(-m^{\Theta(1)}\right).
\end{align*}
In the same way, using de Moivre--Laplace integral limit theorem, we conclude that the probability
that there is no $x\in U$ with $|N_X(x)|\geq n'p+\sqrt{C n'p(1-p)\ln m}$ equals
\begin{align*}
 \left(1-\frac{1+o(1)}{\sqrt{2\pi}}\int_{\sqrt{C\ln m}}^{\infty}e^{-x^2/2}dx\right)^m &\geq
 \left(1-m^{-C/2}\right)^m\\
 &=\exp\left(-(1+o(1))m^{1-C/2}\right)=1-o(n^{-10}).
\end{align*}
Finally, for every  $d\in J(r)$, the probability that there exists at least $C$ nodes in $U$ that have
$|N_X(x)|=d$ is at most
\begin{align*}
 \left(m\frac{1+o(1)}{\sqrt{2\pi n'p(1-p)}}\,m^{-(1-\varepsilon/2)}\right)^C=\Theta\left(m^{-(1/2-\varepsilon/2)C}\right)=O\left(n^{-\varepsilon (1/2-\varepsilon/2)C}\right)=o(n^{-11}).
\end{align*}
The union bound over all $d\in J(r)$ and all $r\in\{0,1,\ldots,C-1\}$ completes the proof of the claim.
\end{subproof}

\begin{claim}\label{cl:from_d_to_d'}
  Let $a,b,b'\in\bZ^C$ be such that, for all $r=0,1,\ldots,C-1$,
  $$
b_r\in J \text{ and } a_r < b_r < b'_r\le b_r+n^{\varepsilon/10}.
  $$
  Then
  \begin{equation}\label{eq:from_d_to_d'}
\mathbb{P}( \mathbf{b}=b \mid \mathbf{a}=a )
\leq(1+o(1))\,\mathbb{P}(\mathbf{b}=b' \mid \mathbf{a}=a )
+\mathbb{P}(\neg\mathcal{E}\mid \mathbf{a}=a ).    
  \end{equation}
\end{claim}

\begin{subproof}
Let $\xi_1,\ldots,\xi_m$ be the random variables defined by \refeq{xi}.
For any permutation $\sigma\in S_m$ and any sequence $(d_1,\ldots,d_m)$, we clearly have
$$
\mathbb{P}(\xi_1=d_{1},\ldots,\xi_m=d_{m})
=\mathbb{P}(\xi_{\sigma(1)}=d_{1},\ldots,\xi_{\sigma(m)}=d_{m}).
$$

Fix an array of non-negative integers $d^r_1\le d^r_2\le\dots\le d^r_{m(r)}$,
where $r=0,1,\ldots,C-1$ and $\sum_{r=0}^{C-1} m(r)=m$. Moreover, using the notation introduced above
for the multiset $\cD=\Mset{d^r_i}_{0\le r<C,\,1\le i\le m(r)}$,
we suppose that $b_r\in J$ and the multiplicity of $b_r$ in $\cD$ equals $\mu(r)\leq C$ for all $r$.
Let $\cD'$ be obtained from $\cD$ by replacing some occurrences of $b_r$ with $b_r+1$. Then
$$
  \frac{\mathbb{P}\left(\bm{\cD}=\cD'\right)}{\mathbb{P}(\bm{\cD}=\cD)}
  =\prod_r\frac{{n'\choose b_r+1}^{\mu(r)}p^{\mu(r)}}{{n'\choose b_r}^{\mu(r)}(1-p)^{\mu(r)}}
  =\prod_r\left(\frac{(n'-b_r)\,p}{(b_r+1)(1-p)}\right)^{\mu(r)}=1\pm O\left(n^{-\varepsilon/2}\sqrt{\ln n}\right).
$$
It follows that
\begin{multline*}
 \mathbb{P}(\mathbf{b}=b' \mid \mathbf{a}=a)
 \geq \mathbb{P}(\mathbf{b}=b'\wedge\mathcal{E}\mid \mathbf{a}=a)
 \geq
  \left(1- O\left(n^{-\varepsilon/2}\sqrt{\ln n}\right)\right)^{n^{\varepsilon/10}}
  \cdot\mathbb{P}(\mathbf{b}=b\wedge\mathcal{E}\mid \mathbf{a}=a)\\
 \geq(1-o(1))\cdot\of{\mathbb{P}(\mathbf{b}=b\mid \mathbf{a}=a)
 -\mathbb{P}(\neg\mathcal{E}\mid\mathbf{a}=a)},
\end{multline*}
completing the proof of~\eqref{eq:from_d_to_d'}.
\end{subproof}

Summing up the inequalities \refeq{from_d_to_d'} over all $b'$, we readily conclude
that, for every two fixed vectors $a$ and $b$ with $a_r < b_r\in J$ for all $r$,
  \begin{equation}\label{eq:degrees_local}
\mathbb{P}( \mathbf{b}=b \mid \mathbf{a}=a )
\leq n^{-\varepsilon C/10} + \mathbb{P}(\neg\mathcal{E}\mid \mathbf{a}=a ).    
  \end{equation}

\newcommand{\smsets}{\mathfrak{D}}
  
It remains to show that this inequality implies Bound \refeq{sXs}.
Speaking of a multiset $\cD$ of size $m$ below, we mean a multiset $\cD=\{d_1,\ldots,d_m\}$
whose integer elements are indexed by $1,\ldots,m$. For such $\cD$, let $\cD'$ denote
the multiset obtained from $\cD$ by removing all elements $d_i$ such that $d_i+|N_{N(u)\cup N(v)}(i)|=|N(y)|$
for some $y\in N(u)\cap N(v)$.
For $s=(s_0,s_1,\ldots,s_{C-1})$, let $\smsets(s)$ denote the family of all multisets $\cD$ of size $m$
such that for every $r$, the sum of distinct elements in $\cD'_r$ is equal to $s_r$, $b_r\in J$, and $\mu(r)\leq C$.
Fix $C$ so that $C>\max\of{22/(\varepsilon(1-\varepsilon)),\,100/\varepsilon}$.
Using Claim~\ref{cl:max-degrees-concentrate} and inequality \refeq{degrees_local}, we infer
\begin{align*}
  \mathbb{P}(\mathbf{s}_X=s)
  &\le\sum_{\cD\in\smsets(s)}\mathbb{P}\left(\bm{\cD}=\cD\right)+\mathbb{P}(\neg\mathcal{E})\\
   &\leq\sum_{\cD\in\smsets(s)}\mathbb{P}\left(\mathbf{b}=b\right)+o(n^{-10})\\
  &=\sum_{\cD\in\smsets(s)}\mathbb{P}\left(\bm{\cD}^\circ=\cD^\circ\right)
  \cdot\mathbb{P}\left(\mathbf{b}=b\mid \bm{\cD}^\circ=\cD^\circ\right)+o(n^{-10})\\
  &=\sum_{\cD\in\smsets(s)}\mathbb{P}\left(\bm{\cD}^\circ=\cD^\circ \right)
 \cdot\mathbb{P}\left(\mathbf{b}=b\mid \mathbf{a}=a\right)+o(n^{-10})\\
 &\leq \sum_{\cD\in\smsets(s)}\mathbb{P}\left(\bm{\cD}^\circ=\cD^\circ \right)
   \cdot\left(n^{-\varepsilon C/10}+\mathbb{P}\left(\neg\mathcal{E}\mid \mathbf{a}=a\right)\right)+o(n^{-10})\\  
&=n^{-\varepsilon C/10}+\sum_{\cD\in\smsets(s)}\mathbb{P}\left(\bm{\cD}^\circ=\cD^\circ \right)
 \cdot\mathbb{P}\left(\neg\mathcal{E}\mid \mathbf{a}=a\right) + o(n^{-10})\\
&=n^{-\varepsilon C/10}+\sum_{\cD\in\smsets(s)}\mathbb{P}\left(\bm{\cD}^\circ=\cD^\circ \right)
 \cdot\mathbb{P}\left(\neg\mathcal{E}\mid \bm{\cD}^\circ=\cD^\circ\right) + o(n^{-10})\\  
 &\leq n^{-\varepsilon C/10}+\mathbb{P}(\neg\mathcal{E})+o(n^{-10})=o(n^{-10}).
\end{align*}

The proof of Lemma \ref{lem:congest-pp}, and therefore of Theorem \ref{thm:congest-pp}, is complete.

\subsection{A first application: Anonymization of Hamiltonian cycle finding}\label{ss:first-application}

\begin{theorem}[Turau \cite{Turau20}]
  There exists a \bcong algorithm $\cA_\mathrm{HC}$ which, for any $p\ge\log^{3/2}n/\sqrt n$,
  asymptotically almost surely finds a Hamiltonian cycle in $G(n,p)$ in $O(\log n)$ rounds.
\end{theorem}

Theorem \ref{thm:congest-pp} implies that the same can be done in the anonymous setting as well.

\hcalgo*

Another application of Theorem \ref{thm:congest-pp} is
considered in the next section.

\section{Anonymous triangle finding}\label{s:triangle}

We now turn to the triangle finding problem.
Recall that triangle \emph{enumeration} can be solved in $\tilde{O}(n^{1/3})$ rounds in the \congest{} model~\cite{chang2024deterministic}. This bound is tight~\cite{izumi2017triangle}, and the lower bound already holds for $G(n,1/2)$.
On the other hand, for triangle \emph{finding}, the best known lower bound is $\Omega(\log\log n)$ rounds~\cite{AssadiS25}.
In particular, it remains unclear if a triangle can be found in $O(\diam(G))$ rounds in the worst-case.

We show that, for a wide range of $p$, triangles can be found asymptotically almost surely on $G(n,p)$ in diameter time, even in the anonymous setting.
In Subsection \ref{ss:triangle-sound}, we give a triangle finding algorithm for the \bcong{} model. The algorithm is also \emph{sound} in the sense discussed in \Cref{s:intro}. The algorithm can be made anonymous using \Cref{thm:congest-pp}. However, the resulting $\sbmodel\short$-algorithm is no longer sound. In Subsection \ref{ss:triangle-nonsound}, we prove that this is indeed necessary: there is no sound $\mb$-algorithm for finding~triangles.

\subsection{A sound \bcong algorithm for triangle finding}\label{ss:triangle-sound}

Formally, we consider the following problem:

\medskip

\noindent
\trifind\,---
Each node $x$ either outputs a label indicating that $x$ belongs to a triangle or halts silently.
If a label is output by some node, then it must be output by exactly three nodes.
Each such triple must form a triangle in the network graph $G$.
In the worst-case setting, at least one triangle must be found unless $G$ is triangle-free.

\medskip

A \bcong algorithm for asymptotically almost surely finding a triangle in $G(n,p)$
is presented in the proof of the following theorem.

\triangle*

Combining Part 1 of Theorem \ref{thm:triangle} with Theorem \ref{thm:congest-pp}, we obtain the following result.

\diametertriangle*

The rest of the section is devoted to the proof of Theorem~\ref{thm:triangle}.

The starting point of our approach can be easily explained by considering the particular case of $p=1/2$.
Assume that a node $u$ of $G(n,1/2)$ was elected as leader.
After receiving a message from $u$, all nodes in $N(u)$ become aware that they are leader's neighbors.
In the next round, every $x\in N(u)$ sends its identifier to all $y\in N(x)$.
If a node $y\in N(u)$ receives an identifier of some node $x$, it detects
a triangle $\{x,y,u\}$ and informs $u$ about this.
In this way, all triangles containing $u$ are detected.
The procedure is sound: if $u$ does not belong to any triangle, it
spreads the failure notifications throughout the network.

The algorithm finds a triangle asymptotically almost surely because every node $u$ in $G(n,1/2)$ is contained in a triangle
asymptotically almost surely.
This is no longer true for small $p$. Indeed, the expected number
of the triangles containing a fixed node of $G(n,p)$ is equal to ${n-1\choose2}\,p^3$
and, hence, tends to 0 if $p=o(n^{-2/3})$. For this reason, the general case of Theorem~\ref{thm:triangle}
is far more complicated than the case of the uniformly random graph.
We need a much finer approach, which is presented in the following lemma.

The \emph{eccentricity} of a node $v$ is defined as the maximum distance from $v$ to another node
and denoted by~$\ecc(v)$.

\begin{lemma}\label{lem:triangle}
  Let $\varepsilon\in(0,1)$ and suppose that $p(n) \ge n^{\varepsilon-1}$ for all $n$.
  Then there exists a positive constant $c=c(\varepsilon)$ such that
  $G(n,p)$ has the following property with probability at least $1-\exp(-n^c)$:
  for every node $u$ there exists a triangle $\{x,y,z\}$ and a value $d\in\{\ecc(u)-2,\ecc(u)-3\}$ such that
  \begin{enumerate}[label=(\roman*)]
  \item
    the distance between $x$ and $u$ is equal to $d$, and
  \item
    $x$ is the unique neighbor of $y$ at distance $d$ from~$u$.
  \end{enumerate}
\end{lemma}

\begin{proof}
Let $\varepsilon\in(0,1)$, $p=p(n) \ge n^{\varepsilon-1}$, and $G_n\sim G(n,p)$. In what follows, we write $\mathcal{S}_r(u)$ to denote the set of nodes at distance exactly $r$ from a node $u$ in $G_n$. We will need the following technical assertion.
\begin{claim}
\label{cl:layers}
With probability at least $1-\exp(-n^{\Theta(1)})$, for every node $u$,
\begin{enumerate}[label=(C\arabic*)]
\item for every $r\geq 1$, if $|\mathcal{S}_r(u)|>\frac{1}{p}$, then $|\mathcal{S}_{r-1}(u)|>\frac{1}{2p(np)}$;\label{c1}
\item if $r=\ecc(u)-2$, then $|\mathcal{S}_r(u)|>\frac{1}{2p(np)}$;\label{c2}
\item if $0\leq r=\ecc(u)-\delta$ for $\delta\geq 3$, then $|\mathcal{S}_r(u)|<\frac{1}{2p(np)^{\delta-3}}$.\label{c3}
\end{enumerate}
\end{claim}

Let us prove the lemma based on Claim~\ref{cl:layers}. Let $N_{\le r}(v)$ denote the set of nodes at distance at most $r$ from the node $v$. By a ball of radius $r$ around $v$ we mean the subgraph of $G_n$ induced on $N_{\le r}(v)$.

Fix a node $u$.  Expose the balls of radius $r$ around $u$ one by one starting from $r=1$. Let $d\geq 0$ be the minimum integer such that $|\mathcal{S}_d(u)|>\frac{1}{2p(np)}$. Claim~\ref{cl:layers} implies that with probability at least $1-\exp(-n^{\Theta(1)})$
$$
\ecc(u)-3\le d\le\ecc(u)-2,
$$
where the first inequality follows from Condition \ref{c3} and the second from Condition \ref{c2}.
Furthermore, Condition \ref{c1} implies that $|\mathcal{S}_d(u)|\leq 1/p$ with probability at least $1-\exp(-n^{\Theta(1)})$. By the Chernoff bound, $|\mathcal{S}_{d+1}(u)|\leq 0.8n$ and, therefore, $|N_{\leq d+1}(u)|\leq 0.9n$ with probability $1-\exp(-\Omega(n))$. Suppose that we have exposed the ball of radius $d$ around $u$. We may assume that
$$
 \frac{1}{2p(np)}<|\mathcal{S}_d(u)|\leq \frac{1}{p}.
$$
It then suffices to prove that, with probability at least $1-\exp(-n^{\Theta(1)})$, the layer $\mathcal{S}_{d+1}(u)$ contains two adjacent nodes $y$ and $z$ such that $y$ has a unique neighbor $x$ in $\mathcal{S}_d(u)$ and $z$ is also adjacent to~$x$.

Let $X$ denote the complement of $N_{\leq d}(u)$. Let $x_1,\ldots,x_m$ be an arbitrary ordering of the nodes in $\mathcal{S}_d(u)$. For every $i=1,\ldots,m$, expose the neighborhood $N_i$ of $x_i$ in $X\setminus N(\{x_1,\ldots,x_{i-1}\})$ and split it into two subsets $N_i^1\cup N_i^2$ of the same size, uniformly at random. For every node $y\in N_i^1$, check whether it has other neighbors in $\mathcal{S}_d(u)$ and if not, then check whether it is adjacent to another node in $N_i^2$. Since at each step of this process, the set $X\setminus N(\{x_1,\ldots,x_{i-1}\})$ has size $\Theta(n)$, the Chernoff bound ensures that, with probability $1-\exp(-\Omega(np))$, each set $N_i$ has size $\Theta(np)$. Therefore, the node $y$ succeeds in fulfilling the above conditions with probability at least $(1-(1-p)^{\Theta(np)})(1-p)^{|\mathcal{S}_d(u)|}$. Then, the probability that all nodes fail in this process is bounded by
\begin{align*}
\biggl(1-(1-(1-p)^{\Theta(np)})  (1-p)^{|\mathcal{S}_d(u)|}\biggr)^{|\mathcal{S}_d(u)|\Theta(np)}
&\leq
\exp\left(-\Theta(|\mathcal{S}_d(u)|\,np)\,(1-e^{-\Theta(np^2)})\,(1-p)^{1/p}\right)\\
&=
\exp\left(-\Theta(|\mathcal{S}_d(u)|\,np)\,(1-e^{-\Theta(np^2)})\right).
\end{align*}
If $p\geq n^{-1/2}$, then we use the bounds $|\mathcal{S}_d(u)|\geq 1$ and $1-e^{-\Theta(np^2)}\geq 1/2$, getting the probability bound $\exp\left(-\Theta(np)\right)$. If $p\leq n^{-1/2}$, then we use the bounds $|\mathcal{S}_d(u)|\geq \frac{1}{2p(np)}$ and $1-e^{-\Theta(np^2)}=\Theta(np^2)$. This yields the bound $\exp\left(-\Theta((1/p)np^2)\right)=\exp\left(-\Theta(np)\right)$, completing the proof of the lemma.

\medskip

\begin{subproof}
Note first that Condition \ref{c1} is immediate when $p(np)=\omega(1)$, i.e., when $p\gg n^{-1/2}$. Otherwise, fix $u$ and a non-negative integer $r\leq n$, expose the ball of radius $r-1$ around $u$, and assume that $s:=|\mathcal{S}_{r-1}(u)|\leq\frac{1}{2p(np)}$. Then $|\mathcal{S}_r(u)|$ is stochastically dominated by $\mathrm{Bin}(n,1-(1-p)^{s})$. The last random variable has expectation $(1-(1-p)^s)\,n\leq\frac{1}{2p}+O(1)$. The Chernoff bound implies that
$$
 \mathbb{P}(|\mathcal{S}_r(u)|>1/p)\leq\mathbb{P}(\mathrm{Bin}(n,1-(1-p)^{s})>1/p)=\exp(-\Omega(1/p)),
$$
and Condition \ref{c1} follows by applying the union bound over $u$ and~$r$.

\smallskip

Consider the event that some set of size $k\leq\frac{1}{2p(np)}$ has neighborhood of size less than $npk/2$.
Its probability is bounded by
$$
 \sum_{k\leq\frac{1}{2p(np)}}{n\choose k}\,\mathbb{P}\left(\mathrm{Bin}(n(1-o(1)),1-(1-p)^k)<\frac{npk}{2}\right)\leq
 \sum_{k\leq\frac{1}{2p(np)}}(ne/k)^k e^{-\Omega(npk)}
 =e^{-\Omega(np)}.
$$
Therefore, we may assume that for every node $u$ and every $r$ such that $|\mathcal{S}_r(u)|\leq\frac{1}{2p(np)}$, we have $|N_{\leq r}(u)|\leq \frac{1+o(1)}{2p(np)}$.

By the union bound, the probability that $G_n$ contains a node adjacent to all other nodes is at most $n p^{n-1}=\exp(-\Omega(n))$. Therefore, we may assume that every node $u$ has eccentricity $\ecc(u)\geq 2$. Since $\ecc(u)-2\geq 0$, also Condition \ref{c2} is immediate when $p\gg n^{-1/2}$. Assume, therefore, that $p=O(n^{-1/2})$. Fix a node $u$ and a non-negative integer $r\leq n$, expose the ball of radius $r$ around $u$, and assume that $s:=|\mathcal{S}_r(u)|\leq\frac{1}{2p(np)}$. It suffices to prove that $|\mathcal{S}_{r+1}(u)\cup \mathcal{S}_{r+2}(u)|<0.9n$ with probability at least $1-\exp(-n^{\Theta(1)})$. Note that $|\mathcal{S}_{r+1}(u)|>1/p$ with probability at most
$$
 \mathbb{P}(\mathrm{Bin}(n,1-(1-p)^s)>1/p)=\exp(-\Omega(1/p)).
$$
So, we expose $\mathcal{S}_{r+1}(u)$ and assume that $s':=|\mathcal{S}_{r+1}(u)|\leq1/p$. Now, the probability that $|\mathcal{S}_{r+2}(u)|>0.9n$ is at most
$$
 \mathbb{P}(\mathrm{Bin}(n,1-(1-p)^{s'})>0.9n)\leq\mathbb{P}(\mathrm{Bin}(n,1-e^{-1+o(1)})>0.9n)
 =\exp(-\Omega(n)).
$$
Condition \ref{c2} follows by applying the union bound over~$u$.

\smallskip

It remains to prove Condition \ref{c3}. Assume first that $p\gg n^{-1/2}$. Then, with probability $1-\exp(-\Omega(np))$, we have $|\mathcal{S}_1(u)|>\frac{np}{2}\gg\sqrt{n}$ for all $u$. Using the Chernoff bound twice, we conclude that the set $\mathcal{S}_2(u)$ is dominating, i.e., $\ecc(u)\leq 3$, with probability $1-\exp(-n)$. Therefore, it is enough to consider the case $\delta=3$ and $r=0$. If $p=o(1)$, then the assertion is immediate. Otherwise, $\ecc(u)=2$ with probability $1-\exp(-n)$, and therefore there is nothing to prove since in the case under consideration we have $\ecc(u)=3$.

It remains to consider $p=O(n^{-1/2})$. As we already proved above, every sufficiently small set has a neighborhood of size at least $(np)/2$ times larger. Therefore, there exists an integer $r'$ such that
$$
 r'\leq\frac{\ln(1/(4p))}{\ln(np/2)}<\frac{1}{\varepsilon}
$$ 
and with probability $1-\exp(-n^{\Theta(1)})$, we have $|\mathcal{S}_{r'}(u)|>\frac{1}{4p}$ for all $u$. Using the Chernoff bound twice, we conclude that the set $S_{r'+1}(u)$ is dominating, i.e., $\ecc(u)\leq r'+2$, with probability $1-\exp(-1/p)$. Therefore, it suffices to prove \ref{c3} for $r=\ecc(u)-\delta$ for $\delta<\frac{1}{\varepsilon}+2$ (in fact, we will only use the fact that $\delta$ is bounded by a constant).

Fix an integer $3\leq \delta<\frac{1}{\varepsilon}+2$, a node $u$, and a non-negative integer $r\leq n$. Expose the ball of radius $r$ around $u$ and assume that $s:=|\mathcal{S}_{r}(u)|\geq\frac{1}{2p(np)^{\delta-3}}$. It is enough to show that $\ecc(u)\leq r+\delta-1$ with probability at least $1-\exp(-n^{\Theta(1)})$. In the same way as above, using the Chernoff bound $\delta-3$ times, we get that $|\mathcal{S}_{r+\delta-3}(u)|\geq1/(2^{\delta-2}p)$  with probability $1-\exp(-n^{\Theta(1)})$. We then expose one more layer and see that $|\mathcal{S}_{r+\delta-2}(u)|\geq n/2^{\delta}$ with probability $1-\exp(-\Omega(n))$. Finally, we conclude that $|\mathcal{S}_{r+\delta-1}(u)|=n$ with probability $1-\exp(-\Omega(n))$, completing the proof of Condition~\ref{c3}.
\end{subproof}

The proof of Lemma \ref{lem:triangle} is complete.
\end{proof}

In the sequel, the unique identifier of a node $v$ is denoted by~$\id(v)$.

\begin{proof}[Proof of Theorem \ref{thm:triangle}]\mbox{}\\  
  \textit{1.}
  In order to design a sound \bcong algorithm $\cA$ asymptotically almost surely solving \trifind on $G(n,p)$,
  we use the well-known fact that the \leader problem can be solved in $O(\diam(G))$ rounds so that
  \begin{itemize}
  \item
    every node knows its distance to the leader $u$, and
  \item
    the leader $u$ knows its eccentricity $\ecc(u)$---which can be communicated to the other nodes
    and, therefore, we assume that all nodes know~$\ecc(u)$.
  \end{itemize}
  A \congest algorithm accomplishing this is described in \cite[Ch.~5.5--5.6]{HirvonenS},
  and it requires minimal effort to adapt this algorithm to the broadcast communication model. 

  After fulfilling this task, the algorithm $\cA$ sets $d=\ecc(u)-3$, and the next rounds are performed as follows.
  Every node at distance $d$ from $u$ sends its identifier to all of its neighbors.
  Every node $y$ that receives exactly one identifier $\id(x)$ from some of its neighbors,
  sends the pair $(\id x,\id y)$ to all of its neighbors. If a node $z$ receives a pair $(\id x,\id y)$
  in this round and the identifier $\id(x)$ in the preceding round, then it detects the triangle $\{x,y,z\}$.
  In the next two rounds, the same procedure is repeated for $d=\ecc(u)-2$. If a node $z$ detects a triangle,
  it notifies its mates $x$ and $y$ by sending the triple $(\id x,\id y,\id z)$ to all neighbors.
  The nodes receiving such a message become aware that a triangle is found and inform also their neighbors about this
  by sending a message \success. The notification is propagated until it reaches the leader $u$.
  If $u$ is not notified within $\ecc(u)$ rounds, it broadcasts an \alarm message. This message propagates
  through the network, leading to global termination in state \failed.
Lemma \ref{lem:triangle} ensures that the failure probability is bounded by~$\exp(-n^c)$.
  
   \textit{2.}
   The algorithm $\cA$ can be combined with parallel execution of the straightforward \bcong algorithm $\cL$
   listing all triangles within $O(n)$ rounds. Under the conditions imposed on $p$, the diameter
   of $G(n,p)$ is bounded by $O(1/\varepsilon)$ with probability at least $1-1/n$; this follows, e.g., from \cite[part~(i) of Theorem 3.27]{Janson_book}.
   Therefore, the running time of $\cA$ is bounded by $O(1/\varepsilon)$ at least with this probability.
   Since the failure probability of $\cA$ is exponentially small, this implies that
   the combined algorithm solves \trifind in expected time $T_\cA(n)+(\frac1n+\exp(-n^c))\,T_\cL(n)=O(1/\varepsilon)$.
\end{proof}

Finally, one might wonder if the following  natural algorithm could also find a triangle in $G(n,p)$: each node picks two neighbors (say, with the smallest identifiers) and checks if they have a common neighbor.  Unfortunately, this idea fails already for constant $p$, as a deterministic \congest{} algorithm has to work with an adversarial identifier assignment; we sketch the reason why in Appendix~\ref{apx:naive}.

\subsection{Soundness vs.\ Anonymity}\label{ss:triangle-nonsound}

The limitations of anonymous networks are well known. The unsolvability of many natural problems can
typically be established using techniques based on covering maps (see, e.g., \cite[Ch.~7]{HirvonenS} and \cite{Angluin80}).
In particular, these results rule out the solvability by \mb{}-algorithms even with arbitrarily
large expected time. By contrast, solving a problem in the \mb model asymptotically almost surely is a more
feasible goal. It turns out, however, that achieving soundness for such \mb{}-algorithms is often inherently difficult
if at all possible.

A homomorphism from a graph $F$ onto a graph $G$ is called a \emph{covering map} if it is a bijection
in the neighborhood of each node of $F$.

The following highly useful observation due to Angluin~\cite{Angluin80} is proved by a simple induction on~$r$.

\begin{lemma}\label{lem:cover}
  Let $x\in V(G)$ and $y\in V(H)$, and suppose that $\alpha(x)=y$ for a covering map $\alpha$
  from $G$ onto $H$. For every \mb{}-algorithm $\cA$ and for every integer $r\ge1$, if $\cA$ is run on $G$ and $H$,
  then the processors $x$ and $y$ are in the same state in round $r$ and, in particular, send identical messages
  to their neighbors in that round.
\end{lemma}

\begin{example}\label{ex:triangle-free}
  The 6-node graphs $G$ and $H$ shown in \Cref{fig:cover}
  admit covering maps from the same graph $F$.
  It follows from \Cref{lem:cover} that the property of a graph being triangle-free
  cannot be detected by any algorithm in the model~\mb.
\end{example}

\tikzstyle{node}=[circle,draw,fill=Black,inner sep=2pt]
\tikzstyle{edge} = [draw,thick,-]

\begin{figure}
\centering
\begin{tikzpicture}%
  \begin{scope}
    \foreach \pos/\name in {{(1,.5)/a}, {(0,1.5)/b}, {(-1,.5)/c},
                            {(-1,-.5)/d}, {(0,-1.5)/e}, {(1,-.5)/f}}
        \node[node] (\name) at \pos {};
    \foreach \source / \dest in {a/b, b/c, c/d, d/e, e/f, f/a, b/e}
       \path[edge] (\source) -- (\dest);
    \node[draw=none,fill=none] at (-1,-1.5) {$G$};
  \end{scope}
  \begin{scope}[xshift=40mm]
    \foreach \pos/\name in {{(0,.5)/a}, {(1,1.5)/b}, {(-1,1.5)/c},
                            {(0,-.5)/d}, {(-1,-1.5)/e}, {(1,-1.5)/f}}
        \node[node] (\name) at \pos {};
    \foreach \source / \dest in {a/b, b/c, c/a, a/d, d/e, e/f, f/d}
       \path[edge] (\source) -- (\dest);
    \node[draw=none,fill=none] at (-1.7,-1.5) {$H$};   
      \end{scope}
  \begin{scope}[xshift=90mm]
    \foreach \pos/\name in {{(.5,.5)/a}, {(0,1.5)/b}, {(-.5,.5)/c},
      {(-.5,-.5)/d}, {(0,-1.5)/e}, {(.5,-.5)/f},
      {(2.5,.5)/aa}, {(2,1.5)/bb}, {(1.5,.5)/cc},
                            {(1.5,-.5)/dd}, {(2,-1.5)/ee}, {(2.5,-.5)/ff}}
        \node[node] (\name) at \pos {};
        \foreach \source / \dest in {a/b, b/c, c/d, d/e, e/f, f/a,
        aa/bb, bb/cc, cc/dd, dd/ee, ee/ff, ff/aa, b/bb, e/ee}
      \path[edge] (\source) -- (\dest);
    \node[draw=none,fill=none] at (-1,-1.5) {$F$};      
  \end{scope}      
\end{tikzpicture}
\caption{Graphs $G$ and $H$ with a common cover graph $F$.}
\label{fig:cover}
\end{figure}
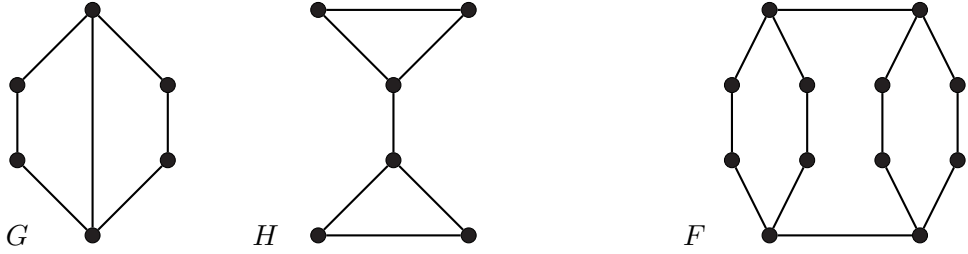

We now show that even a simpler, decision version of \trifind does not admit any sound \mb{}-algorithm.
Specifically, we define the \tridetect problem as follows: If an input graph $G$ contains a triangle
and only in this case, at least one processor outputs~\yes.

\begin{theorem}\label{thm:nosoundtridetect}
  \tridetect cannot be solved asymptotically almost surely by any sound \mb{}-algorithm.
\end{theorem}

\begin{proof}
Assume for the sake of contradiction that
$\cA$ is a sound \mb{}-algorithm for \tridetect. Sample $G \sim G(n,1/2)$.
  A.a.s., $G$ contains triangles and following $\cA$, a.a.s.{}
  at least one of the nodes will terminate in the accepting state \yes.
  Fix $G$ for which this holds.
  Consider the \emph{bipartite double cover} of $G$ and denote it by $G'$.
  For each node $x$ in $G$, the graph $G'$ contains two clones $x'$ and $x''$ of $x$.
  If $x$ and $y$ are adjacent in $G$, then $x'$ and $y''$, as well as $y'$ and $x''$, are adjacent in $G'$.
  Define a map $\alpha\map{V(G')}{V(G)}$ by setting $\alpha(x')=\alpha(x'')=x$
  for all $x\in V(G)$. Note that $\alpha$ is a covering map from $G'$ onto $G$.
  By \Cref{lem:cover}, the algorithm $\cA$ terminates on $G'$ with at least one node
  in state \yes. It remains to note that the bipartite double cover $G'$ is triangle-free.
  This is a contradiction, because not all nodes terminate in state \failed,
  and the output is not a correct solution.  
\end{proof}

Theorem \ref{thm:nosoundtridetect} shows that, although quite desirable, soundness turns out to be
too restrictive for $\mb$-algorithms. In particular, Corollary \ref{cor:diametertriangle}
cannot be improved to achieve soundness while remaining in the anonymous setting.
This negative result is fairly typical. For example, basically the same argument
applies to prove that also \leader does not admit a sound \mb{}-algorithm.
Putting a little more effort, we can prove that the existence of a Hamiltonian
cycle in a graph also cannot be recognized by any sound anonymous algorithm
and, therefore, Corollary \ref{cor:hcalgo} also cannot be strengthened to ensure soundness. 
Specifically, we define the \hamilton problem as follows: If an input graph $G$ contains a Hamiltonian cycle
and only in this case, at least one processor has to output~\yes.

\begin{theorem}\label{thm:nosoundhamilton}
  \hamilton cannot be solved asymptotically almost surely by any sound \mb{}-algorithm.
\end{theorem}

\begin{proof}
  The bipartite cover used in the proof of Theorem \ref{thm:nosoundtridetect} is a particular
  instance of the more general concept known as \emph{lift} (or \emph{cover}) of a graph $G$.
  In general, a lift of $G$ is constructed as follows. For $k\ge1$, consider node-disjoint
  copies $G_0,G_1,\ldots,G_k$ of $G$. Let $u$ and $v$ be adjacent vertices of $G$,
  and $u_i$ and $v_i$ be their copies in $G_i$. In the disjoint union of $G_0,G_1,\ldots,G_k$,
  the copies of the edge $\{u,v\}$ form a matching between the node sets $\{u_0,u_1,\ldots,u_k\}$
  and $\{v_0,v_1,\ldots,v_k\}$. A lift $F$ of $G$ is obtained by replacing this matching, for each $\{u,v\}$,
  with an arbitrary, possibly the same, matching between these node sets.
  Note that a map $\alpha\map{V(F)}{V(G)}$ defined by $\alpha(x_i)=x$
  for all $x\in V(G)$ and $i=0,1,\ldots,k$ is a covering map from $F$ onto $G$.

  The proof idea is the same as in the proof of Theorem \ref{thm:nosoundtridetect}:
  we need a lift of $G$ which does not contain Hamiltonian cycles.
  It is known \cite{LuczakWW15} that a random graph $G$ has the following property:
  if $k$ is large, then almost all lifts of $G$ do contain a Hamiltonian cycle.
  This, however, does not make our task complicated.

  Indeed, assume that a graph $G$ contains a triangle $\{a,b,c\}$ and at least one additional vertex.
  Denote the edges of this triangle by $e_1=\{b,c\}$, $e_2=\{a,c\}$, and $e_3=\{a,b\}$. 
  Consider the disjoint union of four copies $G_0,G_1,G_2,G_3$ of $G$.
  Construct a lift $F$ of $G$ as follows:
  for each edge $e_i=\{x,y\}$, where $i=1,2,3$, we replace the edges $\{x_0,y_0\}$ and $\{x_i,y_i\}$
  with two new edges $\{x_0,y_i\}$ and $\{x_i,y_0\}$, which we will call \emph{twisted}.
  Thus, $F$ has six twisted edges. Note that any Hamiltonian cycle in $F$ must
  contain all six of them. Indeed, for $e_i=\{x,y\}$, such a cycle $C$ must contain
  the twisted edges $\{x_0,y_i\}$ and $\{x_i,y_0\}$ because they provide the only
  possibility for $C$ to enter and exit the part $G_i$. Note also that each of
  the vertices $a_0,b_0,c_0$ is incident to two twisted edges; for example,
  $a_0$ is incident to $\{a_0,c_2\}$ and $\{a_0,b_3\}$. This, however, means
  that once $C$ enters the part $G_0$, which is only possible at one of the vertices $a_0,b_0,c_0$,
  it must exit this part immediately, without visiting any other vertices.
  This contradiction shows that $F$ does not have any Hamiltonian cycle.
  
  We are now prepared to prove the theorem. Assume that $\cA$ is a sound \mb{}-algorithm for \hamilton.
A.a.s., $G(n,1/2)$ has the following properties:
\begin{itemize}
\item
  there exists a triangle $\{a,b,c\}$,
\item
  there exists a Hamiltonian cycle (see \cite[Ch.~8]{Bollobas_book} for combinatorial background
  and \cite{Turau20} for algorithmic aspects),
\item
  $\cA$ successfully detects the Hamiltonicity of $G(n,1/2)$, that is,
  at least one of the nodes terminates in the accepting state~\yes.
\end{itemize}
  Fix $G$ satisfying all three conditions and construct its lift $F$ as described above.
  By \Cref{lem:cover}, the algorithm $\cA$ terminates on $F$ with at least one node
  in state \yes. This is a contradiction, because $F$ does not contain any Hamiltonian cycle.  
\end{proof}

\section{Collapse of average-case \local/\mb time hierarchies}\label{s:collapse}

We now consider anonymous networks with unbounded communication.
We will establish two collapse results:
\begin{itemize}
\item
  All problems are solvable on $G(n,p)$ asymptotically almost surely in \mb using a small number
  of communication rounds. If $p\ge n^{\varepsilon-1}$, then this is even possible in \sbmodel with a constant number of rounds.
\item
  If the processors know the network size $n$, then, on the uniformly random graph $G(n,1/2)$, all distributed
  problems are solvable by sound \sbmodel{} algorithms in just four rounds. Four rounds are optimal in this result.
\end{itemize}
We prove these collapse theorems in \Cref{ss:collapse,ss:compute-all}, respectively.
The proof is based on a close relation between \mb model and the color refinement algorithm,
which is explained in \Cref{ss:CR}.

\subsection{Color Refinement}\label{ss:CR}

On an input graph $G$, the color refinement (CR) algorithm iteratively computes a sequence $C_r$
of colorings of the node set $V(G)$.
The initial coloring is defined by node degrees: $C_0(x)=\deg x$
for all $x\in V(G)$. For each $r>0$, the coloring $C_r$ is defined inductively as
\begin{equation}
  \label{eq:defCr}
C_{r}(x)=\Mset{C_{r-1}(y)}_{y\in N(x)},    
\end{equation}
where $\Mset{}$ denotes a multiset and $N(x)$ is the neighborhood of a node $x$.
A simple induction on $r$ shows that
\begin{equation}
  \label{eq:refines}
\text{if }C_r(x)\ne C_r(x'),\text{ then }C_{r+1}(x)\ne C_{r+1}(x'),
\end{equation}  
that is, each subsequent color partition of the node set $V(G)$ refines the preceding one.

As immediately seen from \refeq{defCr}, the coloring $C_r$ can be computed in \mb in $r$ rounds:
in each round, the nodes send their neighbors all the information they have obtained so far.
That is, in the first round each node sends its degree, and in each subsequent round,
it sends the multiset of all messages received from its neighbors in the preceding round.
This is an \mb analog of a \emph{full-information protocol},
and we refer to it in the sequel as the \emph{universal \mb algorithm}.

Color refinement can be used to characterize \mb-equivalent pairs of nodes with the next lemma.

\begin{lemma}\label{lem:cr-states}
   Let $x\in V(G)$ and $y\in V(H)$. 
  \begin{enumerate}[label=\normalfont \textbf{\arabic*}.]
  \item
Assume that $C_r(x)=C_r(y)$.
  If an \mb{} algorithm $\cA$ is run on $G$ and $H$, then nodes $x$ and $y$
  are in the same state up to the $r$-th round.
\item
  If $C_r(x)\ne C_r(y)$, then there is an \mb{} algorithm $\cA$ with the following property:
  When $\cA$ is run on $G$ and $H$, then nodes $x$ and $y$ enter distinct states
  in the $r$-th round.
  \end{enumerate}
\end{lemma}

\begin{proof}
  \textit{1.}
  We use a close relationship between CR and the concept of universal cover.
Given a node $x \in V(G)$, let $U_x(G)$ denote the unfolding of $G$ into a tree starting from $x$.
Formally, $U_x(G)$ can be defined as the graph whose nodes are non-backtracking walks in $G$,
where two walks are adjacent if one of them is a one-node extension of the other.
The unrooted version $U(G)$ of this tree is called the \emph{universal cover} of $G$ because
$U(G)$ covers every cover $H$ of $G$.
It follows from Lemma \ref{lem:cover} that if $U_x(G)$ and $U_y(H)$ are isomorphic as rooted trees,
then $x$ and $y$ are in the same state, whatever \mb{} algorithm is run on $G$ and $H$.
This can be stated in a fairly useful, more general form.

For an integer $r\ge0$ and a rooted tree $T_x$, let $T^r_x$ be the truncation of $T_x$ at depth~$r$.
Let $U_x=U_x(G)$ and $W_y=U_y(H)$ and suppose that $U^r_x$ and $W^r_y$ are isomorphic rooted trees.
Even this weaker assumption ensures that if an \mb algorithm $\cA$ is run on $G$ and $H$, 
then the processors $x$ and $y$ are in the same state up to the $r$-th round.

It remains to note that $U^r_x$ and $W^r_y$ are isomorphic if and only if $C_r(x)=C_r(y)$.
A proof of this well-known equivalence can be found, e.g., in \cite[Lemma~2.3]{KrebsV15}.

\smallskip

  \textit{2.}
The second claim follows by considering the universal \mb{} algorithm.
\end{proof}

If the $C_r$ colors of all nodes
in a graph $G$ are pairwise distinct, then they can be used to present $G$ in a canonical form.
The approach succeeds a.a.s.{} on the uniformly random graph $G=G(n,1/2)$ already for $r=1$, by the  result of Babai, Erd\H{o}s, and Selkow \cite{BabaiES80}. Their algorithm works as follows:
\begin{itemize}
\item
  Set $m=\lceil3\log_2n\rceil$ and consider the set $U$ consisting of the nodes with the $m$ largest degrees in $G$.
  It is known that the $m$ largest degrees are a.a.s.{} unique, i.e., $|U|=m$ a.a.s.~\cite{Bollobas81}.
\item
  For each node $v$, compute its label as the bit string of length $m$ consisting
  of the adjacencies of $v$ to the nodes in the set~$U$.
\end{itemize}
As proven in \cite{BabaiES80}, these strings are a.a.s.{} pairwise distinct.

This canonical labeling scheme can be straightforwardly translated into a two-round \sbmodel{} algorithm that asymptotically almost surely generates unique identifiers for all nodes.
In the first round, each node $v$ sends its degree to all neighbors. After the message exchange,
$v$ knows the set of the degrees of all of its neighbors.
In the second round, $v$ sends the $m$ largest degrees it has heard about to its neighbors.
Since the diameter of $G(n,1/2)$ is a.a.s.{} equal to 2, every processor learns
the set $U$ and, hence, is able to compute its identifier.
Note that while the identifiers have logarithmic size, the algorithm requires sending
messages of superlogarithmic length $\Omega(\log^2n)$.

Another \sbmodel implementation of the Babai-Erd\H{o}s-Selkow algorithm works in just one round,
requires sending messages of only logarithmic size, and produces identifiers of size $\Theta(\log^2n)$.
Having received the set of their neighbor degrees, each node $v$ chooses the $m$ largest elements
of this set and outputs its ordered list as its identifier $\id(v)$. These identifiers are pairwise distinct a.a.s.{}
because, as easily seen, $\id(v) \ne \id(v')$ whenever $v$ and $v'$ have distinct adjacency patterns to~$U$.
Along with part 1 of Lemma \ref{lem:cr-states}, this has the following consequence, which we state for further reference.

\begin{proposition}[Babai, Erd\H{o}s, and Selkow \cite{BabaiES80}]\label{prop:BES}
  If CR is run on $G(n,1/2)$, then a.a.s.{} $C_1(v)\ne C_1(v')$ for all $v\ne v'$.
\end{proposition}

Subsequent work on CR on $G(n,p)$ is summarized in \Cref{table:summary} in \Cref{s:intro}.
Below, we use the following result from \cite{CzP,GaudioRS25}: if CR is run on $G(n,p)$
with $\frac{(1+\delta)\ln n}{n}\le p\le\frac{1}{2}$, then a.a.s.{} $C_2(v)\ne C_2(v')$ for all $v\ne v'$.

\subsection{Collapse of the time hierarchy for $\local\aas=\mb\aas$}\label{ss:collapse}

We now show that for $p$ over the connectivity threshold,
all distributed problems can be solved on $G(n,p)$
asymptotically almost surely by anonymous algorithms in the \mb model.
\Cref{thm:congest-pp} allows us to obtain  a similar result for the weaker model \sbmodel as well.
In both cases, the time hierarchy collapses to the level equal to the typical diameter of $G(n,p)$ up to a small additive constant.
Recall that for a model $\m$, we write $\m\gnp\aas[r]$ to denote the class of problems solvable on $G(n,p)$
asymptotically almost surely by an \m{} algorithm within $r$ rounds.

\collapse*

\begin{proof}
  \textit{1.}
  For a problem $\Pi$, we design an \mb{} algorithm $\cA$ solving $\Pi$ on $G \sim G(n,p)$ asymptotically almost
  surely in $\diam(G)+3$ rounds.
  The theorem will follow because if $np\to\infty$, then a.a.s.{} $\diam(G(n,p))=(1+o(1))\frac{\ln n}{\ln(np)}+1$;
  see \cite[Theorem 6]{ChungL01}. In what follows, we use the formal definition of a distributed graph
  problem $\Pi$ from Section~\ref{ss:definitions}.
  
  We describe how $\cA$ works on an input graph $G$. The message exchange is actually the same for all
  $\Pi$, and the specific definition of $\Pi$ is used only
  to compute the final states of the processors in the last round.
  The algorithm starts by executing the first three rounds of the universal \mb{} algorithm.
  \begin{itemize}
  \item As shown in \cite{CzP,GaudioRS25}, after the second round, every node $x$ has a
  unique identifier~$C_2(x)$.
\item After the third round, node $x$ knows the identifiers $C_2(y)$ of all its neighbors $y\in N(x)$.
  The CR color $C_3(x)=\Set{C_2(y)}_{y\in N(x)}$, computed in this round, can be seen as the
  neighborhood list of $x$, where each neighbor $y$ of $x$ is represented by its identifier~$C_2(y)$.
  \item In the fourth round, node $x$ sends its neighborhood list $C_3(x)$ to all neighbors,
    and during the remaining $\diam(G)-1$ rounds, $C_3(x)$ is propagated throughout the network.
    It is important to note that, due to \refeq{refines}, the identifier $C_2(x)$ of $x$
    is uniquely determined by~$C_3(x)$.
  \end{itemize}
  In the end, every node knows the neighborhood lists of all nodes in $G$ and
  is able to construct the isomorphic copy $G'$ of $G$ on the set $\Set{C_2(x)}_{x\in V(G)}$ of nodes.
  Finally, each node $x$ represents all solutions
  in $\Pi(G')$ in a natural encoding and sorts them lexicographically.
  Then $x$ chooses the lexicographically least
  solution $\sigma$ and terminates with output $\sigma(C_1(x))$.

\smallskip
  
  \textit{2.}
  The proof is virtually the same with the only difference that, instead of the \mb algorithms
  provided by \cite{CzP,GaudioRS25}, we use the \sbmodel algorithm of our Theorem \ref{thm:congest-pp}.
  Thus, the nodes obtain unique identifiers already in the first round, which not only decreases
  the total number of rounds by one but, more importantly, ensures that also the rest of the algorithm operates in~\sbmodel.
\end{proof}

\subsection{Gathering all information a.a.s.~in $\sbmodel[4]$}\label{ss:gather-all}

In the \local{} model, all nodes can learn the identifiers and structure of the graph $G$ in $\diam(G)+1$ rounds. %
Let $G_v$ denote the rooted version
of $G$ with $v$ designated as root. Thus, by gathering the entire graph $G$, there is a \local algorithm that
allows each processor $v$ to determine the isomorphism type of $G_v$ in $\diam(G)+1$ rounds.

Clearly, a similar result -- in its general form -- is impossible in the \mb model,
even in the average-case setting. Moreover, identifying the network graph $G$ up to isomorphism
is impossible in \mb for any fixed $G$, because $G$ is indistinguishable from any of its lift;
cf.~the proofs of Theorems \ref{thm:nosoundtridetect} and~\ref{thm:nosoundhamilton}.
However, if nodes know the size of $G$, then this task becomes possible a.a.s.{} even in the weaker \sbmodel model.
This observation was at the core of the proof of \Cref{thm:collapse}. We now consider this issue in the case of the uniformly random graph $G(n,1/2)$ in greater detail.

For a graph $G$, let
$$
S_r(G)=\Set{C_r(v)}_{v\in V(G)}
$$
be the set of the node colors produced by the color refinement algorithm in the $r$-th round.
This is an isomorphism invariant of a graph, which is important for us by the following reason.
Assume that $S_r(G)\cap S_r(H)=\emptyset$ for every non-isomorphic graph $H$ with the same number of nodes.
By part 2 of Lemma \ref{lem:cr-states}, this means that in the model \mb, all processors
can identify the isomorphism type of $G$ in $r$ rounds if they know the size of the network.
The proof of part 1 of Theorem \ref{thm:learn5} below shows that this works for $G=G(n,1/2)$
asymptotically almost surely with $r=4$ even in the \sbmodel model.

\begin{theorem}\label{thm:learn5}
  Let $G=G(n,1/2)$. Asymptotically almost surely,
  \begin{enumerate}[label=\normalfont \textbf{\arabic*}.]
  \item
   $S_4(G)\cap S_4(H)=\emptyset$ for every $n$-node graph $H\not\cong G$;
 \item
   $S_3(G)\cap S_3(H)\ne\emptyset$ for some $n$-node graph $H\not\cong G$;
 \item
   $S_2(G)\ne S_2(H)$ for every $H\not\cong G$.
  \end{enumerate}
\end{theorem}

Theorem \ref{thm:learn5} shows how the graph invariant $S_r(G)$ evolves in terms of its ability to identify
the random graph $G \sim G(n,1/2)$.
After two refinement rounds, the \emph{entire} set $S_r(G)$ suffices to identify the graph~$G$.
Nevertheless, there exists a node whose individual color is still insufficient to identify $G$ even after three rounds.
This changes after one more refinement: after four rounds, the color of any single node already identifies~$G$.
To complete the picture, note that the number of rounds $r=2$ in part 3 of the theorem is optimal.
Indeed, as shown in \cite[Lemma 6.2]{PikhurkoV11}, a.a.s.{} there exists $H\not\cong G$ such that $S_1(G)=S_1(H)$.

We also remark that \Cref{thm:learn5} is an average-case analog of the following worst-case result
obtained by Krebs and Verbitsky~\cite{KrebsV15}:
there are $n$-node graphs $G$ and $H$ such that $S_{2n-1}(G)\cap S_{2n-1}(H)=\emptyset$
while $S_r(G)\cap S_r(H)\ne\emptyset$ for all $r\le2n-O(\sqrt n)$. Note that $2n-1$ rounds
are enough to distinguish $G$ and $H$ in \mb unless these graphs are CR-indistinguishable,
and the result of \cite{KrebsV15} shows that this worst-case bound is nearly optimal.

\begin{proof}[Proof of Theorem \ref{thm:learn5}]
  \mbox{}
  
  \textit{1.}
  Assume that $G$ has diameter 2 and all colors $C_1(w)$ for $w\in V(G)$
  are pairwise distinct. Note that  $G(n,1/2)$ has both properties
  a.a.s., the latter due to Proposition \ref{prop:BES}. Thus, the $C_1$ colors
  can be seen as unique identifiers, and we can consider an isomorphic
  copy of $G$ on the set $\Set{C_1(w)}_{w\in V(G)}$, which we denote by $G'$.
  It suffices to show, for an arbitrary $v\in V(G)$, that $G'$ can be reconstructed from the color $C_4(v)$.
  For each $w$, the color $C_2(w)=\Set{C_1(u)}_{u\in N(w)}$
  uniquely determines $C_1(w)$ (see \refeq{refines}) and therefore makes it possible to reconstruct the
  neighborhood of the node $C_1(w)$ in $G'$. Since the diameter is equal to 2,
  the color $C_4(v)$ absorbs (according to the recursive definition in Equation~\refeq{defCr}) the colors $C_2(w)$ of all $w\in V(G)$
  and, therefore, contains all necessary information to reconstruct $G'$ completely.

  \smallskip

\textit{2.}
Assume that $V = \{1, \ldots, n\}$.
Fix a node $v$, say $v=1$.
The random graph $G \sim G(n,1/2)$ on $V$ can be generated in two steps.
First we sample $G_1 \sim G(n-1,1/2)$ on the nodes $\{2,\ldots,n\}$,
and then draw the edge $\{1,u\}$ with probability $1/2$ for all $u\ge2$ independently.
A simple application of Chernoff's bound ensures that a.a.s.{} the degree of every
node in $G_1$ deviates from $n/2$ by at most $\sqrt n\log n$.
It follows that $G_1$ has at least $k=(n-1)/(2\sqrt n\log n)$ nodes
with the same degree $d$. The number of these nodes being connected to $v$
in the second step of generation of $G$ is a binomial random variable
with distribution $\mathrm{Bin}(k,1/2)$. Applying Chernoff's bound once again,
we see that $G$ a.a.s.{} contains a set $U$ of at least $m=(\frac12-o(1))\sqrt n/\log n$
nodes that all have degree $d$ and are not adjacent to~$v$.

Call a quadruple of nodes $x_1,x_2,x_3,x_4$ \emph{variable} if
$\{x_1,x_2\}$ and $\{x_3,x_4\}$ are pairs of adjacent nodes while
$\{x_1,x_4\}$ and $\{x_2,x_3\}$ are not. A \emph{switch} consists
in removing the edges $\{x_1,x_2\}$ and $\{x_3,x_4\}$ from the graph
and adding instead the edges $\{x_1,x_4\}$ and $\{x_2,x_3\}$.
Note that this operation does not change the degree of any node in the graph.
We use the fact that a.a.s.{} every $m$-node subset of $V(G)$ contains
a variable quadruple; see \cite[Lemma 6.2]{PikhurkoV11} for the proof.

Choose a variable quadruple in the set $U$, make a switch, and denote
the modified graph on the set $V$ of nodes by $H$. %
As mentioned above,
\begin{equation}
  \label{eq:eqC1}
  \deg_G(x)=\deg_H(x)\text{ for all }x\in V.
\end{equation}
Here and below we use subscripts or superscripts to specify a graph. Moreover,
\begin{eqnarray}
N_G(x)=N_H(x)\quad&\text{if}&x\notin U,\label{eq:eqN}\\
\deg_G(x)=\deg_H(x)=d &\text{if}&x\in U,\label{eq:eqd} 
\end{eqnarray}
by the definition of $U$ and the construction of $H$.
Along with Eq.~\refeq{eqC1}, Eq.~\refeq{eqN} readily implies that
$C_1^G(x)=C_1^H(x)$ if $x\notin U$. If $x\in U$, then we only know that
$N_G(x)\setminus U=N_H(x)\setminus U$. However, Eq.~\refeq{eqd}
ensures that all nodes in $N_G(x)\cap U$ and in $N_H(x)\cap U$
have the same degree $d$. Therefore,
\begin{equation}
  \label{eq:eqC2}
  C_1^G(x)=C_1^H(x)\text{ for all }x\in V.
\end{equation}
We can see now that $G$ and $H$ are a.a.s.{} non-isomorphic.
Indeed, Proposition \ref{prop:BES} implies that, a.a.s., all colors $C_1^G(x)$ are
pairwise distinct. Along with Eq.~\refeq{eqC2}, this implies
that an isomorphism from $G$ to $H$ must be the identity map of $V$ onto itself.
Since $G$ and $H$ are different graphs, we conclude that no isomorphism can exist.

To complete the proof, it suffices to prove $C_3^G(v)=C_3^H(v)$.
This is equivalent to showing that
\begin{equation}
  \label{eq:eqC3}
C_2^G(w)=C_2^H(w)\text{ for all }w\in N_G(v)=N_H(v).
\end{equation}
Given such $w$, note that $w\notin U$. By \refeq{eqN},
we therefore have $N_G(w)=N_H(w)$. In order to derive Equality \refeq{eqC3},
it suffices to observe that, as follows from \refeq{eqC2},
$C_1^G(x)=C_1^H(x)$ for all $x\in N_G(w)=N_H(w)$.

\smallskip

\textit{3.}
By Proposition \ref{prop:BES}, a.a.s.{} $|S_1(G)|=n$. Let $H\not\cong G$.
If $S_1(G)\ne S_1(H)$, then $S_2(G)\ne S_2(H)$ as well.
Suppose that $S_1(G)=S_1(H)$. Since the graphs are non-isomorphic,
there are nodes $u,u'\in V(G)$ and $v,v'\in V(H)$ such that
$C_1(u)=C_1(v)$, $C_1(u')=C_1(v')$, and $u$ and $u'$ are adjacent while $v$ and $v'$
are not or vice versa. It follows that $C_2(u)\ne C_2(v)$ and, consequently, $C_2(u)\notin S_2(H)$.
\end{proof}

\subsection{Computing everything in $\sbmodel\nsa[4]$}\label{ss:compute-all}

Recall that  $\m\nsa$ denotes the class of problems solvable in the model \m
by non-uniform sound algorithms. We now study the round hierarchy
$$
\m\nsa[1]\subseteq\m\nsa[2]\subseteq\ldots\subseteq\m\nsa[r]\subseteq\m\nsa[r+1]\subseteq\ldots 
$$
for the classes $\m\in\{\mb,\sbmodel\}$.
As discussed in \Cref{sss:collapse},
 we show that this hierarchy collapses
at the fourth level and not before.
Recall that \all denotes the class of all (isomorphism-invariant, component-wise) graph problems as defined in Section~\ref{ss:definitions}.

\allnsacollapse*

\begin{proof}
  \textit{1.}
  For a problem $\Pi$, we have to design a non-uniform sound \sbmodel algorithm $\cA_\Pi$ solving $\Pi$
  on $G(n,1/2)$ in $4$ rounds.
  We proceed similarly to the proof of Theorem \ref{thm:collapse}. A special care is only required to achieve soundness
  by using the knowledge of~$n$.

  The algorithm $\cA_\Pi$ starts by executing the first 4 rounds of the universal (full-information) algorithm for the \sbmodel model.
  More precisely, although we defined the universal algorithm in \Cref{ss:CR} for the \mb model,
  $\cA_\Pi$ executes it in \sbmodel by discarding message multiplicities. For most input graphs, this will make no difference
  as, by Proposition \ref{prop:BES}, the $C_1$-colors of all nodes in $G(n,1/2)$
  are a.a.s.{} pairwise distinct. 

  After the message exchange, each node $v$ knows a set of colors $C_1(w)$
  for some nodes of $G$. In general, $v$ knows these colors not for all $w$
  or, even if so, some of these colors can be equal. In each of these cases,
  the number of pairwise distinct colors $C_1(w)$ seen by $v$ is less than $n$
  and then $v$ immediately terminates with output \failed.

  Assume that the number of these colors is equal to $n$.
  As explained in the proof of part 1 of Theorem \ref{thm:learn5},
  $v$ can now construct an isomorphic copy $G'$ of $G$
  on the set $\Set{C_1(w)}_{w\in V(G)}$, provided that for each $C_1(w)$ it also receives $C_2(w)$. Otherwise, $v$ terminates with output \failed. It also terminates with output \failed if the diameter of $G'$ exceeds~2.

  If $G'$ is constructed and has diameter at most 2, then $v$ finds the lexicographically least solution $\sigma \in \Pi(G')$ and terminates with output
  $\sigma(C_1(v))$.
  It remains to notice that if $|S_1(G)|=n$ and $G$ has diameter at most 2, then
  every $v$ terminates in the state $\sigma(C_1(v))$. On the other hand, if at least one
  of the two conditions is violated, then every $v$ terminates with output \failed.

    \textit{2.}
  It suffices to exhibit a problem $\Pi$ not solvable in $\mb\nsa[3]$.
  We define $\Pi$ as the recognition problem for a certain
  class of graphs $\mathcal{C}$. More precisely, we require that all processors
  terminate in state \yes if $G\in\cC$ and in state \no if $G\notin\cC$.
  Let $\cC_n$ denote the restriction of $\cC$ to the graphs with $n$ nodes.
  The set $\cC_n$ is constructed as follows.
  
  Part 2 of Theorem \ref{thm:learn5} says that, for almost every $n$-node graph $G$
  there exists a non-isomorphic $n$-node graph $H_G$ such that
  $C_3(G)\cap C_3(H_G)\ne\emptyset$. Denote the set of all such $G$ by $\cV_n$.
  Define an auxiliary graph $\cG_n$ on $\cV_n$ by connecting $G$ and $H$ in $\cV_n$
  by an edge if $C_3(G)\cap C_3(H)\ne\emptyset$. Let $\cD_n$ be a minimum dominating set
  in $\cG_n$. Set $\cC_n=\cV_n\setminus\cD_n$. Note that $|\cC_n|\ge|\cV_n|/2$ and, hence,
  the fraction of $\cC_n$ among all $n$-node graphs is at least~$\frac12-o(1)$.

  Assume that $\Pi$ is solvable by a non-uniform sound \mb algorithm $\cA$ in 3 rounds.
  In particular, $\cA$ solves $\Pi$ asymptotically almost surely. This implies that for a
  sufficiently large $n$, the outcome of $\cA$ is correct on some graph $G\in\cC_n$,
  that is, all processors terminate in state \yes. Let $u\in V(G)$ and $v\in V(H)$
  for $H\in\cD_n$ be such that $C_3(u)=C_3(v)$. By part 1 of Lemma \ref{lem:cr-states},
  $v$ terminates in state \yes as well, a contradiction to the soundness assumption.
\end{proof}

As an immediate consequence of Theorem \ref{thm:all-nsa}, we obtain a round hierarchy result.

\begin{corollary}\label{cor:notall4}
$\mb\nsa[3]\subsetneq\mb\nsa[4]$.
\end{corollary}

\section{Conclusions}\label{s:conclusions}

We showed that unique identifiers can be generated on $G(n,p)$ asymptotically almost surely
in $\sbmodel\short[1]$ whenever $n^{\varepsilon-1}\le p\le 1/2$. It is natural to ask
whether the range of $p$ in this result can be made even closer to the connectivity threshold.
In a somewhat weaker form, the question is if unique identifiers can be generated in
$\mb\short[O(1)]$ whenever $(1+\delta)\ln n/n\le p\le 1/2$. We can prove that this is
possible in \mb within $O(\diam(G))$ rounds. Whether this can be done in a constant time
seems to be a challenging question.

Generally, the average-case complexity of distributed problems on $G(n,p)$
deserves further exploration. As we already mentioned in \Cref{s:intro},
the Hamiltonian cycle problem has attracted a considerable attention
\cite{levy2004distributed,chatterjee2018fast,Turau20}.
In \Cref{s:triangle}, we addressed the triangle finding problem.
It is natural to consider also other problems in the average-case setting.
For example, let $p$ be such that $\prob{\diam(G(n,p))\le2}=1/2$. Note that $p=(1+o(1))\sqrt{2\ln n/n}$; see \cite{Bollobas81diam}.
Is there a \congest algorithm determining the diameter of $G\sim G(n,p)$
asymptotically almost surely in a constant number of rounds?
Is there such \emph{sound} algorithm?
Note that in the worst case, distinguishing graphs of diameter 2 from those of diameter 3 requires $\Omega(n)$ rounds, even for randomized algorithms~\cite{HolzerW12}.

While triangle finding is known to have complexity $\Omega(\log \log n)$ in the worst-case~\cite{AssadiS25},
we showed that triangles can be found a.a.s.\ in constant-time on $G(n,p)$ when  $n^{\varepsilon-1}\le p\le 1/2$.
For many other natural subgraphs, such as larger cliques and longer cycles, the detection problem has polynomial-in-$n$ worst-case lower bounds; see e.g.~\cite{drucker2014power,korhonen2017broadcast,czumaj2020detecting,censor2021distributed,fischer2018possibilities}.
Can some other hard-to-detect subgraphs still be found fast on $G(n,p)$ in the \congest{} model?

In \Cref{s:collapse}, we proved that anonymous sound algorithms are very
powerful in the non-uniform setting, when the processors know the network size.
Specifically, all distributed problems solvable in \local
can be solved in this model in just four rounds. This result assumes that the
message size is unbounded. An intriguing question is how powerful (or how weak)
are non-uniform sound algorithms in the $\mb\short$ and $\sbmodel\short$ models
with logarithmically bounded communication.

\bibliographystyle{plainnat}
\bibliography{refs}

\newpage

\appendix

\section{Glossary of distributed computing models}\label{apx:glossary}

In \Cref{table:glossary},
we give a concise glossary of the different models of computing discussed in this work.
We refer the reader to the book of Hirvonen and Suomela~\cite{HirvonenS} for the formal definitions of the \pn{}, \local{} and \congest{} models. The paper of Hella et al.~\cite{HellaJKLLLSV15} gives formal definitions for the weak models \sbmodel, \mb, \svmodel, \mvmodel, and \vvmodel of computing; in this work, we only consider the first two models.

In the \pn{} model, each node $v \in V$ has ports $1, \ldots, \deg(v)$. In each round, each node $v$ sends (distinct) messages to each of its $\deg(v)$ ports, receives a message from each of its ports,
and updates its state based on the vector of $\deg(v)$ received messages in the round.
That is, the state transition can depend on the current state of a node and the vector of received messages.
The \local{} model is simply the port-numbering model \pn{} with unique $O(\log n)$-bit identifiers given as local input to all nodes. The \congest{} is the \local{} model with the additional restriction that messages have length at most $O(\log n)$ bits. That is,
following our notational convention, $\congest = \local\short$.

The weak models can be obtained as restrictions of the \pn{} model; see Hella et al.~\cite{HellaJKLLLSV15} for the formal details.
If the neighbors of node $v$ send the messages $m_1, \ldots, m_{\deg(v)}$ to $v$ in round $r$, then in  $\sbmodel$ the state of node $v$ at the start of round $r+1$ is a function of its current state and the set
\[
\Set{m_ i}{1\le i\le \deg(v)}.
\]
In the $\mb$ model, the new state of node $v$ is a function of its current state and the multiset
\[
\Mset{m_ i}{1\le i\le \deg(v)}.
\]
Note that the assumption that nodes know their degree is necessary in the $\sbmodel$ model.
Otherwise, without any inputs, only constant functions could be computed.
In the $\mb$ model, this assumption could be lifted with the overhead of one additional communication round, as nodes can learn their degree using one round of communication (by counting how many messages they received).

\begin{table}[hb]
  \renewcommand{\arraystretch}{1.2}
  \centering
\begin{tabular}{@{}lllllll@{}}
\toprule
\textbf{Model}           & \textbf{Incoming} & \textbf{Outgoing} & \textbf{Message length} & \textbf{Local input}  &\textbf{Notes}    \\ \midrule
\sbmodel        & Set       & Broadcast  & Any      & Degrees     &  \cite{HellaJKLLLSV15}       \\
$\sbmodel\short$ &  &  & $O(\log n)$  &  &  &  \\
\mb        & Multiset  &            & Any  &            &   \cite{HellaJKLLLSV15}      \\
$\mb\short$ &  &  & $O(\log n)$ &  & & \\
\pn        & Vector    & Vector     &  Any        &  &  \cite{Angluin80,HirvonenS}     \\ \midrule
\local     & Vector          & Vector &  Any         & Unique $O(\log n)$-bit IDs & \cite{Linial1992,HirvonenS}  \\
\congest   &  &            & $O(\log n)$ &  &   \cite{HirvonenS}   \\
\bcong &         & Broadcast  &  &           & \cite{drucker2014power,korhonen2017broadcast}      \\
 \bottomrule
\end{tabular}
\caption{An overview of the main models  discussed in this work. The ``incoming'' column denotes whether the nodes receive a set, multiset, or a vector of messages. The ``outgoing'' column denotes whether the nodes can broadcast a single message to all neighbors or a distinct message to each neighbor. The ``message length'' column indicates any restrictions on the length of the sent messages. For these columns, a blank cell denotes that the entry is the same as in the above cell.
In all models, nodes receive their degree as input. The first five models are anonymous, whereas the last three are non-anonymous. In the non-anonymous models, nodes also receive unique identifiers of length $O(\log n)$ bits as local inputs.}
\label{table:glossary}
\end{table}

We remark that port numbering is not a crucial ingredient for the models \local{} and \congest{}.
Indeed, whenever a node $v$ sends a message, it can attach its identifier $\id(v)$. If $v\in N(u)$,
then $\id(v)$ serves as a pointer to the port connecting $u$ and $v$. This observation is important for simulating \local{} and \congest{} algorithms in the weak models, where ports are not numbered.

\section{A naïve algorithm for \trifind in \congest}\label{apx:naive}

We make note of a natural algorithm idea that one could expect to find a triangle in $G(n,p)$:
\begin{itemize}
\item
Each node $x$ sends its identifier $\id(x)$ to all its neighbors. As a result, after one round, every node knows its neighborhood.

\item
Having received the neighbor identifiers, each node $x$ chooses the smallest two, $\id(y)$ and $\id(z)$,
and sends the triple $(\id y,\id z,\id x)$.
\item
Having received the message $(\id y,\id z,\id x)$, a node $y$ checks if it is adjacent to $z$.
If so, then $\{x,y,z\}$ is a triangle, and $y$ outputs the triple $(\id y,\id z,\id x)$.
\end{itemize}

Recall that the node set of $G(n,p)$ is $V=\{1,\ldots,n\}$. Assume that $\id(x) =x$ for all $x\in V$.
In this case, the above naïve triangle finding algorithm succeeds on $G(n,1/2)$ with probability $1-n^{-\Omega(\log n)}$.
This still does not mean that this algorithm solves \trifind in \congest asymptotically almost surely.
Indeed, on almost all $G$ the algorithm has to succeed for \emph{all} identifier assignments,
in particular, for all $\id\in S_n$, where $S_n=\mathrm{Sym}(V)$ is the symmetric group of all
permutations of $V$. It turns out that this is not the case. Below, we outline a proof for
a constant edge probability $p\in(0,1)$. This is likely true also if $p=o(1)$, but this case
requires much more intricate analysis.

A set $I\subset V$ is called \emph{2-dominating independent} set in a graph
if it is independent and every node from $V\setminus I$ has at least two neighbors in~$I$.

\begin{claim}
\label{cl:ind-dom}
Let $p\in(0,1)$ be a constant. Set $k=\lfloor 1.5\log_{1/(1-p)}n\rfloor$. Then a.a.s.{} $G\sim G(n,p)$
contains a 2-dominating independent set of size~$k$.
\end{claim}
 
\begin{proof}[Proof sketch]
Let $X_k$ be the number of independent sets $I$ of size $k$ in $G$ such that every node that does not belong to $I$ has at least two neighbors in $I$. It suffices to show that a.a.s.{} $X_k\geq 1$. To this end, we will prove
that $\mathbb{E}X_k\to\infty$ as $n\to\infty$. 
By Chebyshev's inequality,
$$
\mathbb{P}(X_k=0)\leq\frac{\mathrm{Var}X_k}{(\mathbb{E}X_k)^2}.
$$
The inequality $X_k\geq 1$ will then follow from the equality $\mathbb{E}X_k(X_k-1)=o((\mathbb{E}X_k)^2)$. 
The proof of the last fact is standard --- see the proofs of similar facts (with almost verbatim computations)
in~\cite{dutta2023induced,mcdiarmid2012largest,kamaldinov2021maximum} --- 
but requires dense computations. Therefore, we omit it.

A fixed set $I$ of size $k$ is 2-dominating independent with probability 
$$
 (1-p)^{{k\choose 2}}\left(1-(1-p)^k-kp(1-p)^{k-1}\right)^{n-k}=
 (1-p)^{{k\choose 2}}\left(1-n^{-3/2+o(1)}\right)^{n-k}.
$$ 
Therefore,
$$
 \mathbb{E}X_k={n\choose k}(1-p)^{{k\choose 2}}\left(1-n^{-3/2+o(1)}\right)^{n-k}=\exp\left(k\ln n-k\ln k+k-{k\choose 2}\ln\frac{1}{1-p}+O(\ln\ln n)\right).
$$
Since
$$
 \ln n-1+\ln k-\frac{k-1}{2}\ln\frac{1}{1-p}=\left(\frac{1}{4}+o(1)\right)\ln n,
$$
we conclude that $\mathbb{E}X_k\to\infty$ with $n$ increasing, as desired.
\end{proof}

Let $I$ be a 2-dominating independent set of size $k$ in an input graph $G$,
which exists a.a.s.{} for $G\sim G(n,p)$ by Claim~\ref{cl:ind-dom}.
Consider $\id\in S_n$ mapping $I$ to $\{1,\ldots,k\}$. Then, for every $x\in V$, its two neighbors
with smallest identifiers belong to $I$ and, hence, are not adjacent. Therefore, the algorithm does not find any triangle.

\section{The \umaxdeg problem}\label{app:UMD}

We here give a simple example showing that the class of problems $\local\sa[O(1)]$
is strictly larger than $\local[O(1)]$, and the same holds true as well for $\bcong\sa[O(1)]$
and $\bcong[O(1)]$. Both separations follow from the fact that
\begin{equation}
  \label{eq:Obig}
\bcong\sa[O(1)]\nsubseteq\local[O(1)].
\end{equation}
To prove \refeq{Obig}, we consider a natural problem,
which is a special version of \leader, namely:

\medskip

\noindent
\umaxdeg\,---
If a processor $u$ is a unique node in $G$ with maximum degree, i.e., $\deg(v)<\deg(u)$
for all other $v\in V(G)$, then $u$ must terminate in state \yes.
If this is not the case, $u$ must terminate in state~\no.

\medskip

Note that the random graph $G(n,1/2)$ contains a unique node of maximum degree a.a.s.{} \cite[Theorem~3.9]{Bollobas-b}.

\begin{claim}
\umaxdeg is in $\bcong\sa[9]$.  
\end{claim}

\begin{subproof}
  To show this, consider the following algorithm.  For a node $v$, let $N[v]$ denote the
  closed neighborhood of $v$, i.e., $N[v]=N(v)\cup\{v\}$.
  The case of a single processor is trivial, and we assume that the algorithm is run
  on a connected graph $G$ with at least two nodes.

  \first Node
  $v$ sends the pair $(\deg v,\id v)$ to the neighbors. After the message exchange,
  $v$ computes the maximum node degree $D(v)=\max_{w\in N[v]}\deg(w)$ and the multiplicity $M(v)$ of
  occurrences of this maximum value in the closed neighborhood of~$v$.
  
  \second Node
  $v$ sends the message
  $$
  \begin{cases}
    (D(v),M(v))&\text{if }M(v)\ge2,\\
    (D(v),M(v),\id w)&\text{if }M(v)=1,\text{ where }w\in N[v]\text{ is such that }D(v)=\deg w.
  \end{cases}
  $$
  A node $v$ is marked as \ambitious if for all $u\in N(v)$, it holds that $\deg v\ge D(u)$,
  where equality can occur only if $M(u)=1$ and $\id(v) =\id(w)$ for the node $w\in N[u]$
  with $\deg w=D(u)$. Note that $v$ receives this flag exactly when it has unique maximum
  degree in its 2-neighborhood $N_2[v]=\bigcup_{u\in N[v]}N[u]$. In this case, no other
  node in $N_2[v]$ is marked by this flag.
  
  \third
  Every \ambitious node $u$ sends the message $(\claiming,\id u)$ to the neighbors.
  For convenience, we suppose that $u$ also receives this message (from itself).
  Note that a node can receive such a message from at most one of the neighbors.

  \fourth
  Every node $v$ which received a message $(\claiming,\id u)$ in the preceding round
  forwards it to all of its neighbors. After the message exchange, if a node $v$
  has received so far no message $(\claiming,\id u)$ at all or at least two messages $(\claiming,\id u)$
  and $(\claiming,\id u')$ with $u\ne u'$, then it is marked as \concerned.

  \fifth
  If a node $v$ is \concerned, then it sends the message \alarm to the neighbors.
  If $v$ is not yet \concerned but received a message $(\claiming,\id u)$ in the preceding round,
  then $v$ forwards this message to all of its neighbors. After the message exchange, if a node $v$
  receives at least two messages $(\claiming,\id u)$ and $(\claiming,\id u')$ with $u\ne u'$
  in Rounds 3--5, then it is marked as \concerned. Thus, the number of \concerned nodes can increase.

  \rounds{6--9}
  The following action is repeated: All \concerned nodes send the message \alarm to their neighbors.
  Once a node receives this message, it becomes \concerned itself.

  \round{9}
  After the message exchange is completed, all nodes enter their final states.
  If a node $v$ is \concerned, then it terminates in state \failed.
  If $v$ is not \concerned, then it terminates in state \yes if it is \ambitious
  and in state \no otherwise.

\medskip
  
  To show that the above algorithm solves \umaxdeg asymptotically almost surely, it suffices to observe that
  a correct solution is produced on every $G$ of diameter at most 2 with unique node $u$
  of maximum degree. Indeed, $u$ becomes \ambitious in the 2nd round, and no other
  node in $V(G)=N_2[u]$ is marked by this flag. In the next two rounds, all nodes
  receive the message $(\claiming,\id u)$ and, hence, none of them becomes \concerned
  in the 4th and all subsequent rounds. Consequently, $u$ terminates in state \yes,
  and all other nodes terminate in state~\no.

  \smallskip

  It remains to argue that the algorithm is sound. We prove that either
  the final state configuration is a correct solution or all processors terminate in state \failed.

  Assume first that no node was marked as \ambitious. Then all nodes become \concerned
  in the 4th round and, consequently, all terminate in state \failed.

  Consider now the case that exactly one node was marked as \ambitious.
  If $N_2[u]=V(G)$, then none of the nodes ever becomes \concerned,
  $u$ terminates in state \yes, and all other nodes terminate in state~\no.
  Otherwise, all nodes in $V(G)\setminus N_2[u]$ become \concerned in the 4th round,
  and the message \alarm is propagated in the next five rounds through the entire 2-neighborhood $N_2[u]$.
  Therefore, all nodes of $G$ terminate in the 9th round in state \failed.

  Finally, assume that the set $A$ of \ambitious nodes has size at least 2.
  All nodes in the set $C=V(G)\setminus\bigcup_{u\in A}N_2[u]$ become \concerned
  in the 4th round. Let $u\in A$. The assumption $|A|\ge2$ implies that $V(G)\ne N_2[u]$.
  Consider a node $w'$ at distance 3 from $u$. If $w'\in C$, then all nodes in $N_2[u]$
  become \concerned in the 9th round at latest. Otherwise $w'\in N_2[u']$ for some $u'\ne u$
  in $A$. Let $w\in N_2[u]\cap N(w')$. If $w$ was \concerned in the 4th round, then all nodes
  in $N_2[u]$ become \concerned in the 8th round at latest. Otherwise, if $w'$ was \concerned in
  the 4th round, then all nodes in $N_2[u]$ become \concerned in the 9th round at latest.
  If neither $w$ nor $w'$ were \concerned in the 4th round, then both of them are \concerned in
  the 5th round and---again---all nodes in $N_2[u]$ become \concerned in the 9th round at latest.
  Therefore, all nodes in $\bigcup_{u\in A}N_2[u]$ eventually become \concerned.
  It follows that all nodes in $V(G)=C\cup\bigcup_{u\in A}N_2[u]$ terminate in state \failed.
\end{subproof}

\begin{claim}
Any \local algorithm $\cA$ for \umaxdeg must make at least $n-6$ rounds in the worst case.  
\end{claim}

\begin{subproof}  
  Indeed, consider the graph $G$ obtained from the path graph $P_{n-2}$ by attaching two
  leaves to one of the endpoints and the graph $H$ obtained from $P_{n-4}$ by attaching two
  leaves to both endpoints. Let $u$ be the node of degree 3 in $G$, and $w$ be one of
  the nodes of degree 3 in $H$. Assume that the identifiers are assigned to the
  nodes in $G$ and $H$ coherently (in particular, $\id(u)=\id (w)$), and run $\cA$ on $G$ and $H$.
  Then $u$ and $w$ will be in the same state in the first $n-6$ rounds, while their final states
  are different.
\end{subproof}
  
\end{document}